\documentclass[letterpaper]{article} 
\usepackage{aaai24}  
\usepackage{times}  
\usepackage{helvet}  
\usepackage{courier}  
\usepackage[hyphens]{url}  
\usepackage{graphicx} 
\usepackage{natbib}  
\usepackage{caption} 
\usepackage{algorithm}
\usepackage{algorithmic}
\usepackage{amsmath,amsthm}
\usepackage{amsmath}
\usepackage{amssymb}
\usepackage{epsfig}
\usepackage{booktabs}
\usepackage{subfigure}
\usepackage{threeparttable}
\usepackage{makecell}
\usepackage{multirow}
\usepackage{color}
\usepackage{IEEEtrantools} 

\newtheorem{theorem}{Theorem}
\newtheorem{definition}{Definition}
\newtheorem{lemma}{Lemma}
\newtheorem{corollary}{Corollary}
\newtheorem{proposition}{Proposition}
\newtheorem{assumption}{Assumption}
\usepackage{IEEEtrantools} 
\def\ie{\textit{i.e.}}
 
\def\eg{\textit{e.g.}}

\usepackage{newfloat}
\usepackage{listings}
\DeclareCaptionStyle{ruled}{labelfont=normalfont,labelsep=colon,strut=off} 
\floatstyle{ruled}
\newfloat{listing}{tb}{lst}{}
\floatname{listing}{Listing}
\title{Revisiting Gradient Pruning: A Dual Realization for Defending against Gradient Attacks}
\author {
Lulu Xue\textsuperscript{\rm1},
Shengshan Hu\textsuperscript{\rm1}\thanks{Corresponding Author.},
Ruizhi Zhao\textsuperscript{\rm1},
 Leo Yu Zhang\textsuperscript{\rm2},
Shengqing Hu\textsuperscript{\rm3},
Lichao Sun\textsuperscript{\rm4},
Dezhong Yao\textsuperscript{\rm5}
}
\affiliations {
    \textsuperscript{\rm 1} School of Cyber Science and Engineering, Huazhong University of Science and Technology\\
\textsuperscript{\rm 2} School of Information and Communication Technology, Griffith University\\
 \textsuperscript{\rm 3} Department of Nuclear Medicine, Union Hospital, Tongji Medical College, Huazhong University of Science and Technology\\

 \textsuperscript{\rm 4} Department of Computer Science and Engineering, Lehigh University\\
    \textsuperscript{\rm 5} School of Computer Science and Technology, Huazhong University of Science and Technology\\
$\{$lluxue,hushengshan,zhaoruizhi,dyao$\}$@hust.edu.cn,  hsqha@126.com, leo.zhang@griffith.edu.au, james.lichao.sun@gmail.com
}
\usepackage{bibentry}

\begin{document}

\maketitle

\begin{abstract}
Collaborative learning (CL) is a distributed learning framework that aims to protect user privacy by allowing users to jointly train a model by sharing their gradient updates only. 
However, gradient inversion attacks (GIAs), which recover users' training data from shared gradients, impose severe privacy threats to CL. 
Existing defense methods adopt different techniques, e.g., differential privacy, cryptography, and perturbation defenses, to defend against the GIAs. 
Nevertheless, all current defense methods suffer from a poor trade-off between privacy, utility, and efficiency.
To mitigate the weaknesses of existing solutions, we propose a novel defense method, Dual Gradient Pruning (DGP), based on gradient pruning, which can improve communication efficiency while preserving the utility and privacy of CL. Specifically, DGP slightly changes gradient pruning with a stronger privacy guarantee. And DGP can also significantly improve communication efficiency with a theoretical analysis of its convergence and generalization. Our extensive experiments show that DGP can effectively defend against the most powerful GIAs and reduce the communication cost without sacrificing the model's utility.
\end{abstract}
\section{Introduction}\label{sec:intro}

Collaborative learning (CL)~\cite{CL} is a distributed learning framework, where multiple users train a model locally and share their gradients among the peers or to a centralized server. 
CL claims to protect user privacy since users do not need to share their local (private) data directly. 
However, recent studies reveal that gradients can be used to recover the original training data information via gradient inversion attacks (GIAs)~\cite{dlg,IVG}.
To against GIAs, a large number of studies have been proposed, where they leverage the advanced privacy protection techniques, such as differential privacy (DP)~\cite{dp4}, cryptography~\cite{encrypt1,encrypt2,encrypt0} and perturbation defense~\cite{ATS,soteria,Precode}. However, none of the existing defense methods could take care of all privacy, utility, and efficiency difficulties in the CL framework. 

For example, traditional defenses such as DP and cryptography-based methods strike a balance among privacy protection, model performance, and efficiency simultaneously ~\cite{dp4,encrypt1, encrypt2,encrypt0}. To address this challenge, various perturbations-based methods have been proposed~\cite{ATS,soteria,Precode}.
But they all rely on auxiliary optimization modules to reduce certain privacy leakage and 
cannot defend against all GIAs in practice (see {Sec.~\ref{sec:privacy}} for details).
For instance, perturbation-based defense methods (\ie, Precode~\cite{Precode}, Soteria~\cite{soteria}) can effectively defend against passive GIAs~\cite{IVG,sapag,framework}, but fail to work against the active GIAs~\cite{curious,exploring}, {which is considered as the state-of-the-art attack method}. 
On the contrary, the classic Top-$k$ based gradient pruning method~\cite{DGC,Top-k} is  generally ineffective for enhancing privacy against passive GIAs, and corresponding defenses (\eg,~Outpost~\cite{outpost}) offer limited protection.
%
But we find that they significantly outperform recent defense methods under the active attack. Tab.~\ref{tab:four attacks} gives a detailed experimental result for this observation. The new findings inspire us to seek a more practical and effective defense against both passive and active GIAs.
In this paper, we propose a new gradient pruning-based method, \textbf{D}ual \textbf{G}radient \textbf{P}runing (DGP). Dual gradient pruning is a novel gradient pruning technique, which removes top-$k_1$ largest gradient parameters and the bottom-$k_2$ smallest gradient parameters from the local model. DGP leads to a strong privacy protection against both passive GIAs and active GIAs.

To measure the level of protection, we present the theoretical analysis of reconstruction error from pruned gradients, showing that the error is proportional to gradient distance. So removing larger gradient parameters can rapidly enlarge the gradient distance, resulting in a significant reconstruction error. However, removing many larger parameters will significantly impact the model's utility. Thus, to improve the pruning ratio, which is essential to robustness against active attack~\cite{curious,robbing}, we also remove smaller gradient parameters. In this way, our method could significantly mitigate GIAs without affecting the model's utility.

We conduct extensive experiments to evaluate our method. The quantitative and visualized results show that our design can effectively make recovered images unrecognizable under different attacks, and reduce the communication cost. 
Our contributions are as follows: 1) We revisit gradient pruning to show its potential for mitigating GIAs; 2) We propose an improved gradient pruning strategy to provide sufficient privacy guarantee while balancing the model accuracy and the system efficiency; 3) We conduct extensive experiments to show that our design outperforms existing defense methods \textit{w.r.t.} privacy protection, model accuracy, and system efficiency.
\section{Related Work}\label{sec:relatedwork}
Collaborative learning~\cite{CL} is considered to be a privacy-preserving framework for distributed machine learning as the training data is not directly outsourced. However, the emerging of GIAs~\cite{dlg,a1,idlg,IVG,can,curious,see,R-GAP,robbing} shatters this conception. It has been proven that the attacker (\eg, a curious server) can easily recover the private data from gradient to a great extent. 
The privacy guarantee of collaborative learning urgently needs to be strengthened.

\textbf{Traditional Defense.}~Traditionally, there are two approaches to construct privacy-preserving collaborative learning: using DP to disturb gradients~\cite{dp4,dp2,dp0,dp3,dp1} or using cryptographic tools to perform secure aggregation~\cite{encrypt3,encrypt1,encrypt2,encrypt4,shuffle1,encrypt0}.
DP~\cite{dp4} is a popular and effective privacy protection mechanism by adding random noise to the raw data, but it is well known that the noise introduced by DP can greatly degrade the model accuracy when meaningful privacy is enforced~\cite{dp5}. Cryptographic-based secure aggregation can guarantee both privacy and accuracy simultaneously, but it incurs expensive computation and communication costs~\cite{op}. Using the shuffle model~\cite{shuffle0,shuffle1} can only provide anonymity. Moreover, it totally changes the system model of collaborative learning since an additional semi-trusted third party is introduced to work cooperatively with the server.

\textbf{Perturbation Defense.}~Recently, researchers have begun to explore the possibility of constructing new gradient perturbation mechanisms to better balance privacy and accuracy. \cite{soteria} proposed Soteria, a scheme that perturbs the representation of inputs by pruning the gradients of a single layer.
 \cite{ATS} proposed ATS, an optimized training data augmentation policy by transforming original sensitive images into alternative inputs, to reduce the visibility of reconstructed images.
 \cite{Precode} presented Precode to extend the model architecture by using variational bottleneck (VB)~\cite{VB} to prevent attackers from obtaining optimal solutions to reconstructed data.
{
These works focus on the semi-honest setting~\cite{dlg,dlg0,framework} but fail to protect privacy when an active server modifies the model to launch GIAs~\cite{robbing}.}
Moreover, these works suffer from high computation costs or a huge communication burden. 


\textbf{Gradient Pruning Defense.}~From an independent research domain, gradient pruning has been commonly used for  saving communication bandwidth.
The most common pruning strategy is Top-$k$ selection, which retains top $k$ gradient parameters with the largest absolute values \cite{DGC,Top-k}. 
It has been widely proved that gradient pruning provides very limited privacy protection ability~\cite{dlg,ATS,evaluating,soteria,Precode} unless a high pruning ratio (\eg, removing 99\% of the gradients) is used at the cost of 10\% accuracy drop~\cite{evaluating}. 
However, we emphasize that this is misunderstood as they only consider the Top-$k$ selection strategy and it has never received an in-depth investigation in the field of security. It is originally designed for improving system efficiency, thus a direct application inherently suffers from many weaknesses. 
{Recently, \cite{outpost} proposed Outpost, a privacy-preserving method that combines Top-$k$ gradient pruning with adaptive noise addition. However, our experiments indicate that Outpost cannot  effectively defend against passive GIAs.} In contrast, our work shows that a slight modification can unleash the potential of gradient pruning to provide a strong privacy guarantee, as shown in Sec.~\ref{sec:OurMethod}.
\section{Threat Model and Gradient Attacks}\label{Sec:ThreatModel}
In this work, we consider a strong threat scenario, where an active server, after receiving gradients from users, tries to reconstruct the local training data and is motivated to modify model parameters in each iteration to strengthen the attack effect. 
As will be shown in Sec.~\ref{Sec:Theory} and Sec.~\ref{experiments}, our method provides a theoretical guarantee against passive attacks and empirical protection against active attacks. So we briefly discuss both kinds of attacks below. 

\noindent\textbf{Analytical Attack (Passive).} Analytical attack exploits the structure of the gradients to recover the inputs, such as using gradient bias terms~\cite{a0}. Recently proposed R-gap attack~\cite{R-GAP} exploits the recursive relationship between gradient layers to solve the input. An effective analytical attack depends on the specific structure and parameters of gradients.

\noindent\textbf{Optimization Attack (Passive).} Optimization attack is firstly proposed in \cite{dlg}, which approximates the desired data $(\mathbf{x}, \mathbf{y})$ with dummy data  $(\mathbf{x^*},\mathbf{y^*})$ by  optimizing 
 the euclidean distance between the gradients $\textbf{g}^*$ (generated by dummy data $(\mathbf{x^*},\mathbf{y^*})$) and the original gradients $\nabla \mathbf{W}$ (produced by real private data $(\mathbf{x}, \mathbf{y})$) with L-BFGS optimizer.
{\cite{IVG} proposed IG, optimizing the cosine distance with Adam optimizer, and \cite{see} proposed GI, optimizing the Euclidean distance with Adam optimizer. These methods are state-of-the-art optimization attacks. Furthermore, recent works~\cite{deep,GGL} utilize GANs to generate data approximating the input. 
However, these attacks are impractical as they necessitate training GANs with vast amounts of data that closely resemble private data.}

Despite different optimizers can be used to achieve better attack quality~\cite{IVG,sapag,framework},  the existing attacks are all measured by the distance between the virtual gradients $\textbf{g}^*$ and the original gradients $\nabla \mathbf{W}$.
We therefore propose a general definition for passive attacks to better evaluate their performance. From Definition~\ref{Def:attack}, for a given  success probability $(1-\delta)$,  a smaller $\varepsilon$ indicates a better attack strategy $\mathcal{A}$.
\begin{definition}
\label{Def:attack}
A passive attack $\mathcal{A}$ is a $(\varepsilon,\delta)$-passive attack, if it satisfies:
\begin{equation}
\mathbb{P}(\mathbb{E}(\mathcal{D}_\mathcal{A}(\nabla \mathbf{W}, \mathbf{g}^*)) 
\le \varepsilon )\ge 1-\delta.
\end{equation}
where $\mathbb{P}$ represents the probability, $\mathbb{E}$ represents the expectation,  $\mathcal{D}_\mathcal{A}$ is the distance  (commonly instantiated with Euclidean or cosine distance) estimated under $\mathcal{A}$.
\end{definition}

\noindent\textbf{Active Server Attack.}~In this kind of attack, the server can  actively modify the global model  to realize a better attack result rather than honestly executing the protocols~\cite{curious,exploring,fishing}. 
Recently proposed   Rob  attack~\cite{robbing}  adds imprint modules to the model and uses the difference between the gradient parameters in adjacent rows of the imprint module to recover the data, achieving the best attack effect in the literature.
\section{Dual Gradient Pruning }\label{sec:OurMethod}
\subsection{Analysis of Gradient Pruning}
\label{Sec:rob_ivg}
We owe the failure of common Top-$k$ gradient selection methods to two reasons: 1) the distance between the Top-$k$ pruned gradient {$\mathbf{g}$} and the real gradient $\nabla \mathbf{W}$ is small; and 2) large gradient parameters in $\nabla \mathbf{W}$ also reveal label information about user data. 
\begin{figure}[!t]
  \centering
  \subfigure[PSNR ($\downarrow$)]{
    \includegraphics[width=0.2\textwidth]{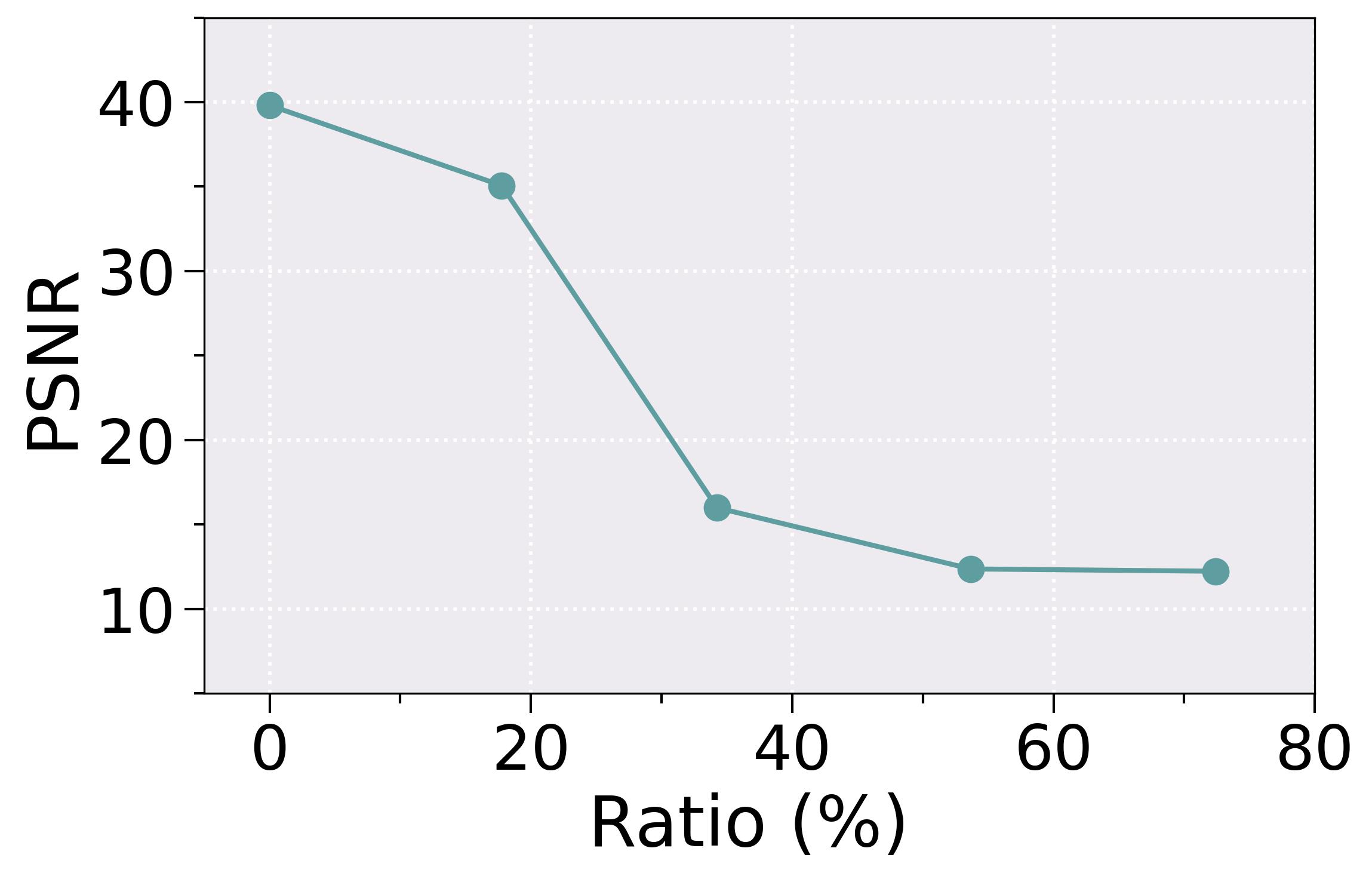}}
  \subfigure[MSE ($\uparrow$)]{
   \includegraphics[width=0.2\textwidth]{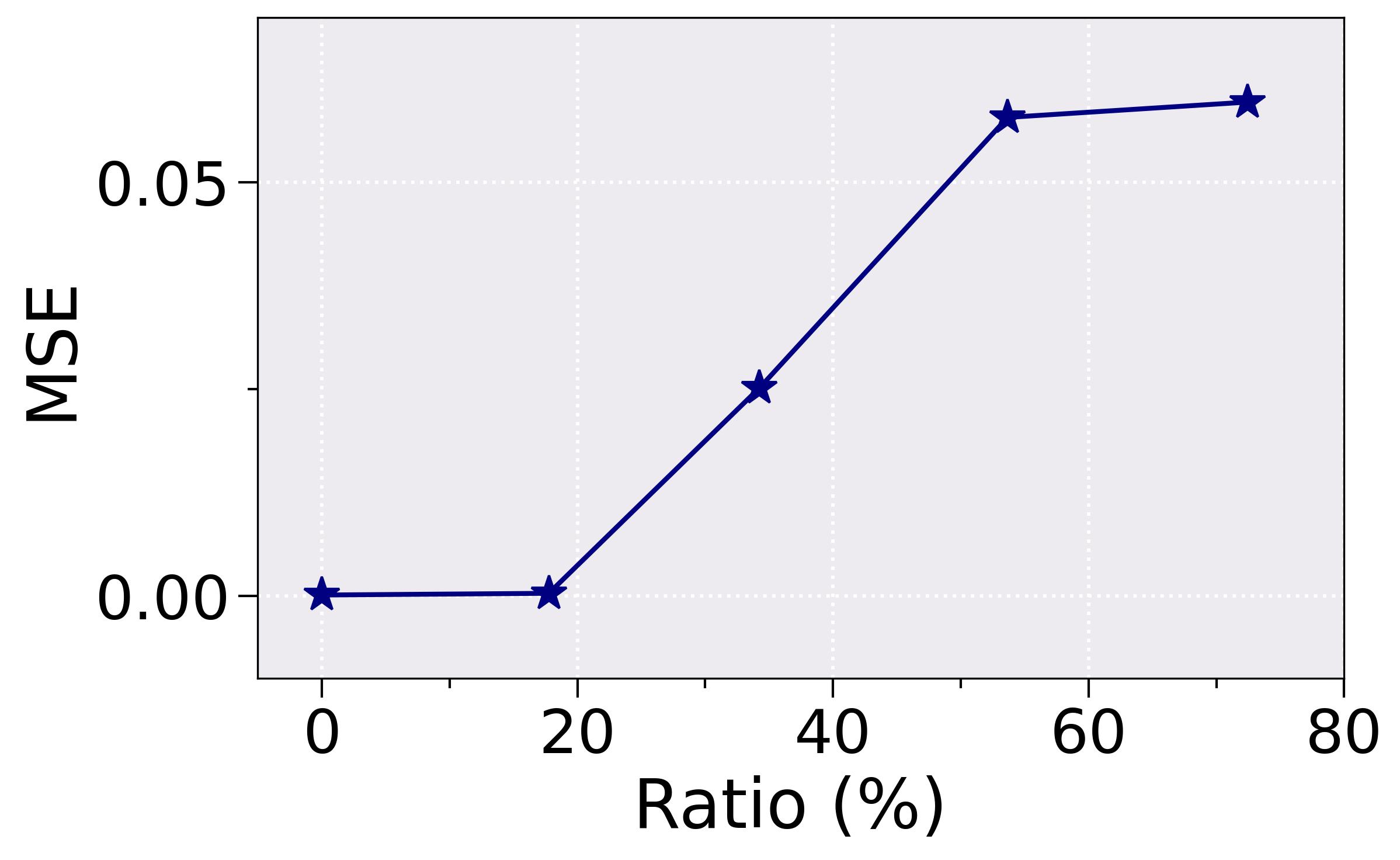}}
\subfigure[SSIM ($\downarrow$)]{
   \includegraphics[width=0.2\textwidth]{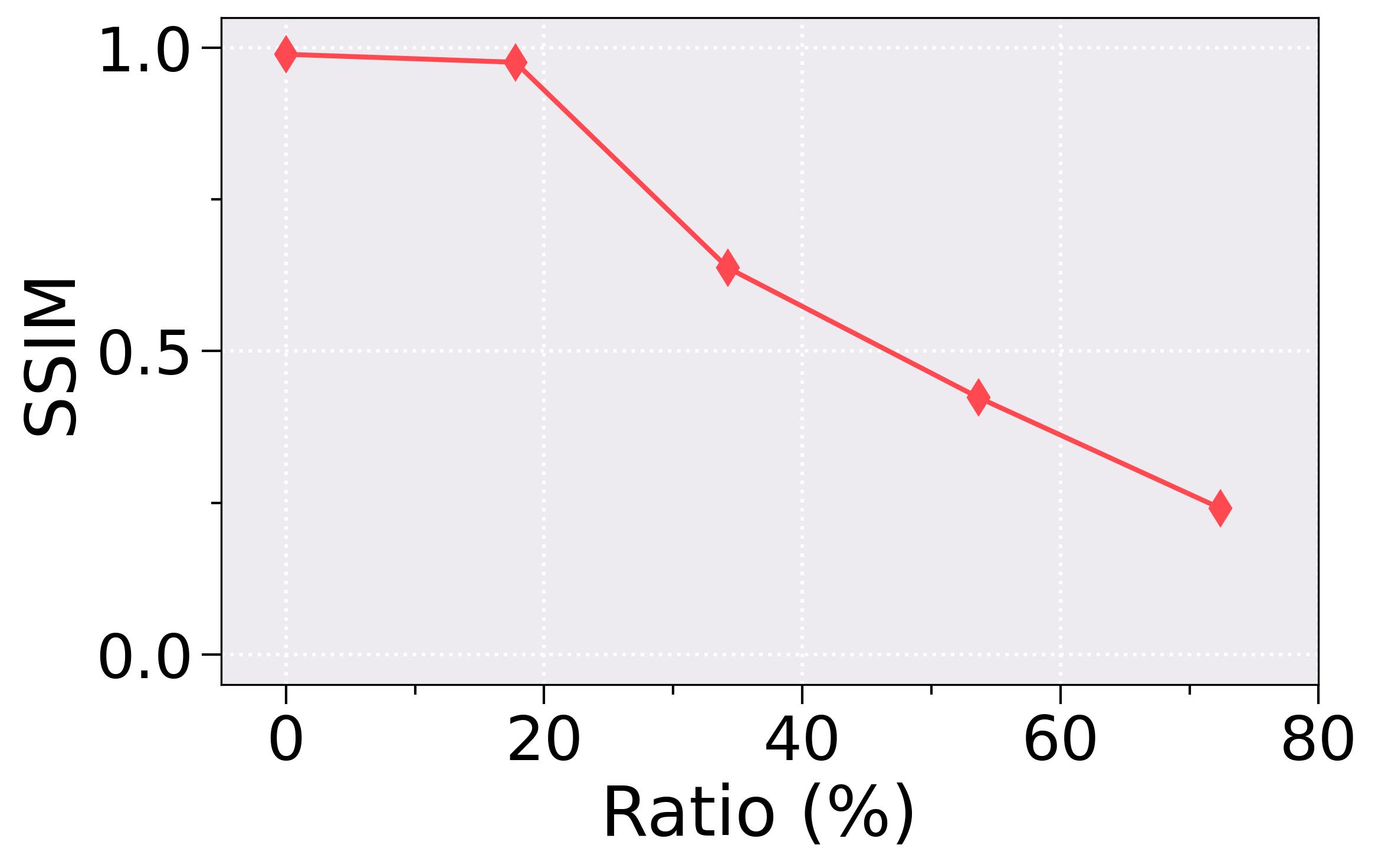}}
\subfigure[LPIPS ($\uparrow$)]{
   \includegraphics[width=0.2\textwidth]{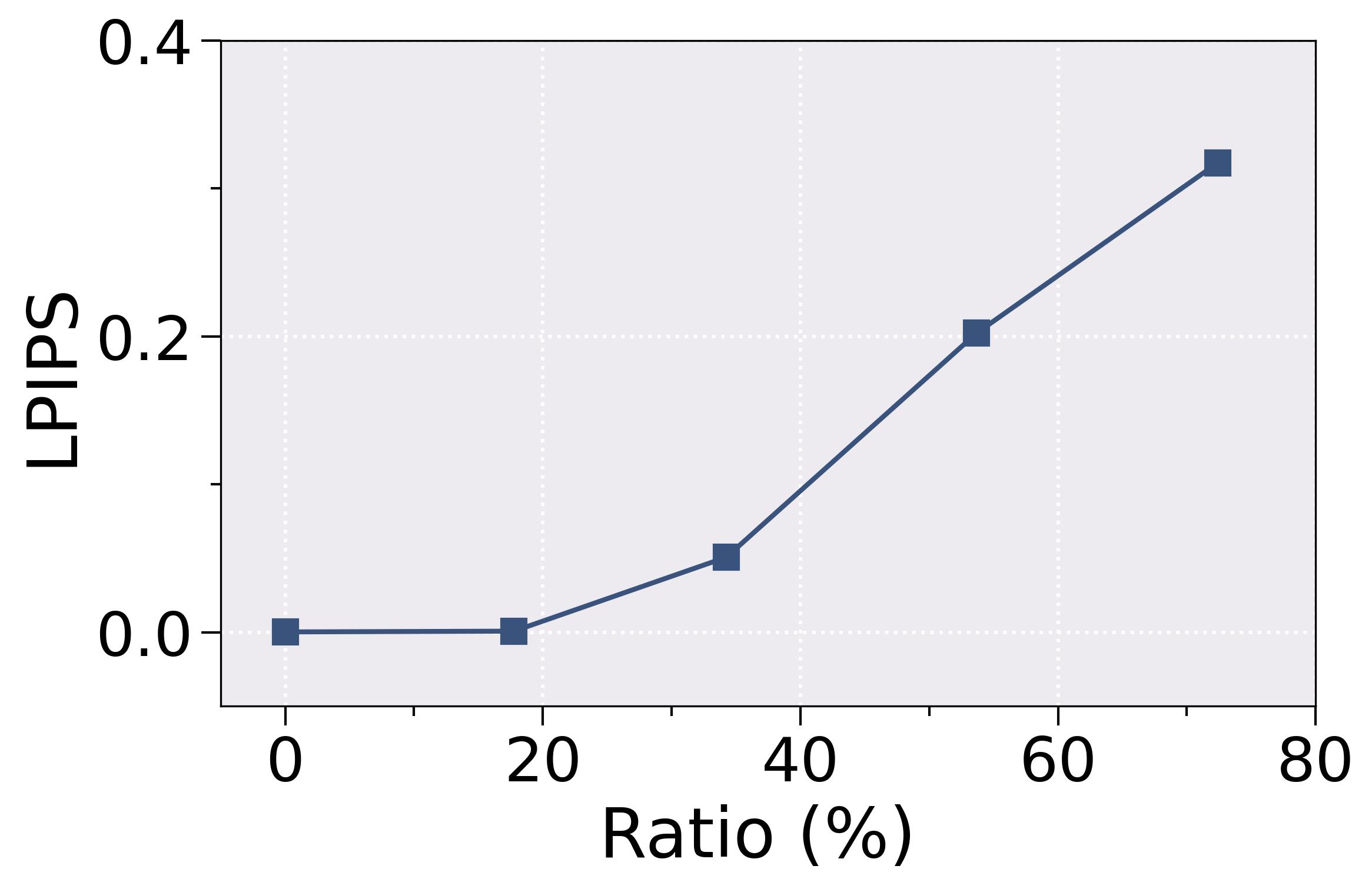}}
\subfigure[Visualization of original and reconstructed data at various ratios]{\includegraphics[width=0.41\textwidth]{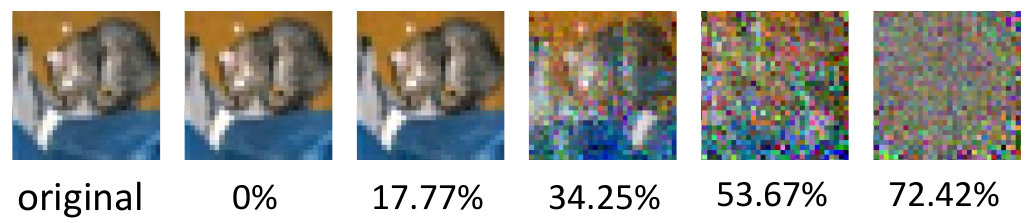}}
   \caption{Relationship between relative gradient distance and reconstruction quality under IG (CIFAR10\cite{cifar} with ResNet18~\cite{resnet18}).}
        \label{fig:ratio}
        
\end{figure}
{The first reason stems from the intuitive observation that when the perturbed gradient is close to the true gradient, it becomes easier for the attacker to infer sensitive information about the true gradient.  And we give a specific example to illustrate this point. In particular, Fig.~\ref{fig:ratio} plots the recovery results of IG attack (in terms of PSNR ($\downarrow$), MSE ($\uparrow$), LPIPS ~\cite{lpips}~($\uparrow$), SSIM~\cite{ssim}~($\downarrow$) metrics) under various relative gradient distance ${||\nabla \mathbf{W}- \mathbf{g}||_2}/{||\nabla \mathbf{W}||_2}$ (measured in ratio). It is clear from the figure that greater distance leads to worse reconstruction for all metrics.
To better support this observation, we propose the following  non-rigorous proposition.}

\begin{proposition}
\label{theo:errorbound}
For any given input $\mathbf{x}$ and shared model $\mathbf{W}$, the distance between the recovered data  $\mathbf{x'}$ and the real data $\mathbf{x}$ is bounded by:
\begin{equation}
||\mathbf{x} -\mathbf{x'}||_2  \ge \frac{||\varphi (\mathbf{x},\mathbf{W})- \varphi (\mathbf{x'},\mathbf{W})||_2}
{||\partial \varphi (\mathbf{x},\mathbf{W})/\partial\mathbf{x}||_2},
\end{equation}
where {$\varphi$ is the mapping from input to the gradient}, \ie, the reconstruction quality is limited by $||\varphi (\mathbf{x},\mathbf{W})- \varphi (\mathbf{x'},\mathbf{W})||_2=||\nabla\mathbf{W }- \mathbf{g}  ||_2$.
\end{proposition}

{Referring to the proof technique of Lemma 1 in \cite{soteria}, we employ the first-order Taylor expansion in our proof. The specific proof of the above proposition is moved to the {appendix} due to space limit (the same hereinafter). And we will present a more rigorous analysis in our follow-up study.} According to the above example and this proposition, it is clear that the reconstruction error is proportional to the gradient distance $||\nabla\mathbf{W} - \mathbf{g}||_2$, i.e., effective defense methods should enlarge the gradient distance as much as possible. 
However, for the Top-$k$ gradient selection {\cite{DGC,Top-k}}, {the $k$ largest parameters are retained, making the gradient distance small by nature. }
To explain the second reason, we consider a $L$-layer perceptron model trained with cross-entropy loss for classification. Let a column vector $\mathbf{r} =[r_1,r_2,…,r_n]$ be the logits (the output of the $L$-th linear layer) that input to the softmax layer, 
the confidence score probability vector is thus $\left[ \frac{e^{r_1}} {\sum_{j}e^{r_j}}, \frac{e^{r_2}} {\sum_{j}e^{r_j}}, \cdots, \frac{e^{r_n}} {\sum_{j}e^{r_j}} \right]$ and the succinct form of the cross-entropy loss becomes $\ell(\mathbf{x},y) =-\log (\frac{e^{r_y}} {\sum_{j}e^{r_j}})$. Focus on the $L$-th layer  $\mathbf{W}^L\mathbf{x}+\mathbf{b}^L=\mathbf{r}$, it is easy to find
\begin{IEEEeqnarray*}{rCL}
\frac{\partial \ell(\mathbf{x} ,y)}{\partial b_i} = 
\frac{\partial \ell(\mathbf{x} ,y)}{\partial r_i} \cdot \frac{\partial r_i}{\partial b_i} =\frac{\partial \ell(\mathbf{x} ,y)}{\partial r_i} =
\frac{e^{r_i}}{\sum_{j}{e^{r_j}} }-\mathbb{I}_{i=y}, 
\end{IEEEeqnarray*}
and
\begin{IEEEeqnarray*}{rCL}
\nabla \mathbf{W}^L =
\frac{\partial \ell (\mathbf{x} ,y)}{\partial r} \cdot \mathbf{x}^T 
= [\frac{\partial \ell(\mathbf{x},y)}{\partial r_1},\cdots,\frac{\partial \ell(\mathbf{x},y)}{\partial r_n}] \cdot \mathbf{x}^T.
\end{IEEEeqnarray*}
For a given $\mathbf{x}$ (and so $\mathbf{x}_t$ is fixed), the magnitude of certain elements of the gradient matrix  $\nabla \mathbf{W}^L$ ({i.e., the $i$-th row}) is {particularly large} if $i$ is the true label of the training data $\mathbf{x}$ due to reason that $|\frac{\partial \ell (\mathbf{x}, y)}{\partial r_i}|=\sum_{j\ne i}|\frac{\partial \ell (\mathbf{x}, y)}{\partial r_j}|$.
   
To summarize, due to the above two reasons, we conclude that common Top-$k$ gradient selection cannot provide sufficient protection for user data against passive optimization attacks.
From another point of view, a sufficient gradient pruning ratio also plays an important role in defending against active server attacks. 
As mentioned in Sec.~\ref{Sec:ThreatModel}, active attackers can exploit the correspondence of partial gradient parameters to recover the real data. So, the gradient pruning will directly destroy the relationship among gradient parameters constructed by the active attacker. Intuitively, the higher the pruning rate, the stronger the impact. As will be validated in Sec.{\ref{experiments}}, a {high} pruning rate can prevent the attacker from obtaining useful gradient information.
\begin{figure}
  \centering
  \subfigure[Recovered data under IG]{
\label{fig:IVG-Topk}
    \includegraphics[width=0.35\textwidth, height=0.15\textwidth]{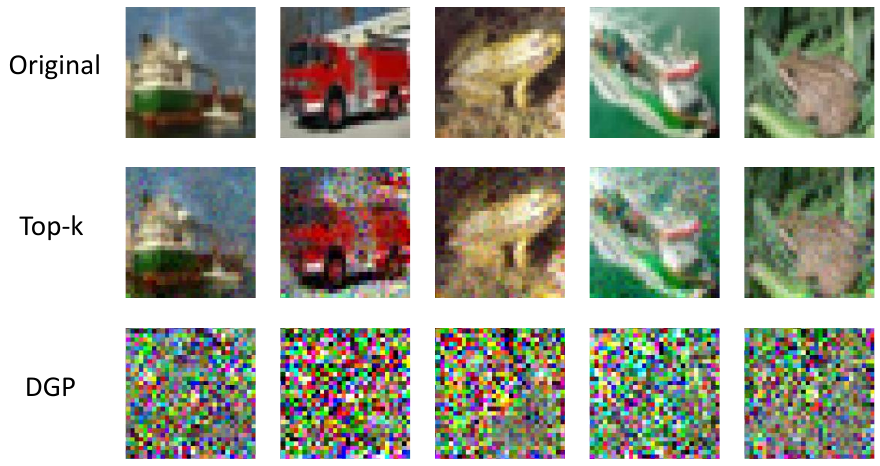}}
  \subfigure[ResNet18 on CIFAR dataset]{
  \label{fig:ACC_topk_ours}
 \includegraphics[width=0.35\textwidth, height=0.21\textwidth]{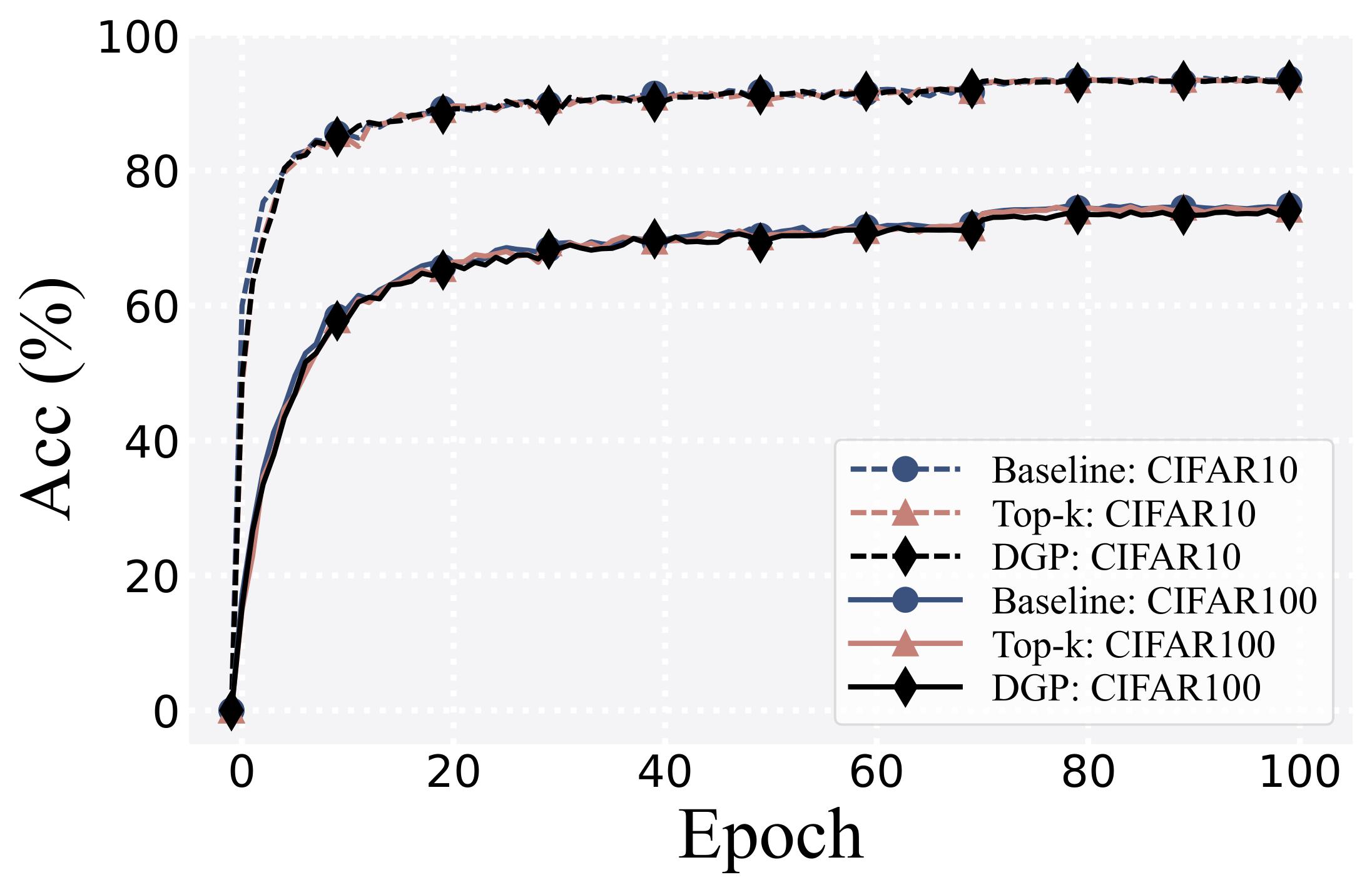} }
\caption{Comparison between Top-$k$ and DGP on privacy and accuracy 
 ($20\%$ of parameters are selected in Top-$k$).}
\label{fig:topk_ours}\end{figure}

\subsection{Dual Gradient  Pruning}
\label{sec:4.2}

Generally speaking, large gradient parameters of local model need to be removed to make the gradient distance larger, but the distance should also be appropriately bounded to maintain high model accuracy.
{ Moreover, it is also necessary to delete gradient parameters to achieve a high pruning ratio, which can reduce the input information that the active server may retain on the gradient by modifying the model and improve communication efficiency. Considering the model performance, we choose to remove small gradient parameters to achieve this.}
With these observations, we propose dual gradients pruning {(DGP)}, a new parameter selection strategy for gradient pruning.
The users first {layer-wisely} sort the absolute values of 
local gradient parameters $\nabla \mathbf{W}$ in the descending order. 
Let $\mathcal{T}_{k_1}(\mathbf{\nabla \mathbf{W}})$ represent the set of top-$k_1$ {percents} of elements of $\mathbf{\nabla \mathbf{W}}$, $\mathcal{B}_{k_2}(\mathbf{\nabla \mathbf{W}})$ represent the set of its bottom-$k_2$ {percents}. 
Then the users remove $\mathcal{T}_{k_1}(\mathbf{\nabla \mathbf{W}})$ and $\mathcal{B}_{k_2}(\mathbf{\nabla \mathbf{W}})$ from $\mathbf{\nabla \mathbf{W}}$  for gradient pruning.
{A detailed illustration of DGP is shown in Alg.~\ref{alg: defense}.} Note that we set $p=k_1/k_2$ as a hyperparameter to regulate the trade-off between privacy and accuracy. 
The authors in \cite{DGC} show that large gradient parameters are more likely to have an impact on the model's performance, hence removing these large parameters will reduce model’s accuracy. To reduce this negative impact and increase convergence speed, we introduce the error feedback mechanism \cite{EF}. In particular, at the iteration round $t$, after user $i$ obtaining his local gradient $\nabla \mathbf{W}_{t,i}$, he will combine $\nabla \mathbf{W}_{t,i}$ with an error term accumulated in the previous ($t-1$) rounds before performing the DGP. A complete illustration of our method is shown in Alg.~\ref{sls}, {and the steps from $\mathbf{e}_{t,i}$ to $\mathbf{e}_{t+1,i}$ provide the implementation details of error feedback mechanism.}
We emphasize that although such dual gradients pruning strategy is very simple, it can significantly mitigate GIAs without affecting the model accuracy. {Fig.~\ref{fig:IVG-Topk} gives an example of ResNet18 showing the privacy guarantee when $k_1=5\%, k_2=75\%$. 
Fig.~\ref{fig:ACC_topk_ours} gives a comparison of model performance.} The convergence analysis of our method is shown in Sec.~\ref{Sec:Theory}, and more experimental results can be found in Sec.~\ref{experiments}.
\begin{algorithm}[t]
\caption{Dual Gradient  Pruning (DGP).}
\label{alg: defense}
\begin{algorithmic}[1] 
\REQUIRE ~~\\ 
  Original gradient matrix $\nabla \mathbf{W }$, values of $k_1$ and $k_2$.
 \FOR{$l \leftarrow 1$ to $L$}
\STATE Search sets $\mathcal{T}_{{k_1 }}(\nabla \mathbf{W }^{l})$ and $\mathcal{B}_{{k_2 }}(\nabla \mathbf{W }^{l})$.\\
\STATE Obtain $\mathbf{g}^{l}$ by removing the parameters in $\mathcal{T}_{{k_1 }}(\nabla \mathbf{W }^{l})$ and $\mathcal{B}_{{k_2 }}(\nabla \mathbf{W }^{l})$
from $\nabla \mathbf{W }^{l}$.\\
 \ENDFOR{}
\STATE return Pruned gradient matrix $\mathbf{g}=\left\{ \mathbf{g}^i\right\}_{i=1}^L$.
\end{algorithmic}
\end{algorithm}
\begin{algorithm}[t!]
\caption{A Complete Illustration of Our Defense.}
\label{sls}
\begin{algorithmic}[1] 
\REQUIRE ~~\\ 
 Initial  model $\mathbf{W}_{0}$, value $k_1$ and $k_2$, total rounds $T$, total users $N$ .
 \STATE  Set $\mathbf{e}_{0}=0$.
 \FOR{$t \leftarrow 0$ to $T-1$}
\FOR{$i \leftarrow 1$ to $N$}
\STATE The $i$-th user generates local gradient $\nabla \mathbf{W }_{t,i}$.
\STATE  $\mathbf{P}_{t,i} =\nabla \mathbf{W }_{t,i} +\mathbf{e}_{t,i}$.
\STATE  $\mathbf{g}_{t,i} =\textnormal{DGP} (k_1 , k_2, \mathbf{P} _{t,i})$
\STATE  $\mathbf{e}_{t+1,i} =\mathbf{P} _{t,i}-\mathbf{g} _{t,i}$
 \ENDFOR{}
 \STATE Sever side aggregation:
 \STATE
 $\mathbf{W}_{t+1}=\mathbf{W}_{t}-\eta\frac{\sum_{i=1}^N \mathbf{g}_{t,i}}{N}$ 
\ENDFOR{}
\STATE return Shared global model $\mathbf{W}_{T}$.
\end{algorithmic}
\end{algorithm}

\section{Theoretical Analysis}\label{Sec:Theory}
This section presents the security analysis with regard to passive GIAs,  
as well as the {generalization} and convergence analyses of our method. 

\subsection{Assumptions}
Following the literature studies in \cite{marginal,EF}, for a given $L$-layer centralized model, we model the first ($L-1$) layers as a robust feature extractor of any input sample. Thus, the function of this model is characterized by $f(x | \mathbf{W}) = \mathbf{W} x + \mathbf{b}$, and the optimization objective is the loss
$\ell(\mathbf{x},y)$ (such as cross-entropy). 
To facilitate analyses and following literature studies \cite{scalecom,hyper,EF}, the assumptions 
about the smoothness of DGP and $l$, as well as the variance of the stochastic gradient are employed.
\begin{assumption}
\label{Assumption:A3}
The pruning mechanism $\textnormal{DGP}(k_1, k_2, \cdot)$ is Lipschitz, so the following conditions  hold:
\begin{IEEEeqnarray}{rCl}
&&||\nabla \mathbf{W} -\textnormal{DGP}(k_1, k_2, \nabla \mathbf{W})  ||_2^2 
\nonumber\\
&&= ||\textnormal{DGP}(0, 0, \nabla \mathbf{W})-\textnormal{DGP}(k_1, k_2, \nabla \mathbf{W})  ||_2^2  \le \gamma_1 ||\nabla \mathbf{W}||_2^2,\nonumber
\end{IEEEeqnarray}
{where $\gamma_1$ is a constant related to $k_1$ and $k_2$ and satisfies $(1-\sqrt{1-k_1*k_2} )^2<\gamma_1<1$.} 
\end{assumption}
\begin{assumption}
The objective function $l: R^d\to R$ has a low bound $l^*$ and it is Lipschitz-smooth, \ie, for any $x_1$, $x_2$, $||\nabla l(x_1)-\nabla l(x_2)||_2\le K||x_1-x_2||_2$ and $l(x_1) \le l(x_2)
+ \langle \nabla l(x_2), x_1-x_2 \rangle 
+ \frac{K}{2}||x_1-x_2||^2_2$.
\label{Assumption:A1}
\end{assumption}
\begin{assumption}
The collaborative stochastic gradient $\nabla \mathbf{W}_{t,i}$ $(t= [0, T-1], i= [1, N])$  is bounded, \ie, $||\nabla \mathbf{W}_{t,i}||^2_2\le G^2$, {and the average aggregated gradient $\nabla \mathbf{W}_{t}$ is the expectation of  collaborative stochastic gradient $\nabla \mathbf{W}_{t,i}$, \ie, $\nabla \mathbf{W}_t=\mathbb{E}(\nabla \mathbf{W}_{t,i} )$}.
Moreover, the variance between $\nabla \mathbf{W}_{t,i}$ and $\nabla \mathbf{W}_t $ is bounded: $\mathbb{E}||\nabla \mathbf{W}_{t,i} - \nabla \mathbf{W }_t||^2_2\le \sigma ^2$.
\label{Assumption:A2}
\end{assumption}
\subsection{Security Analysis}

When considering passive attacks, we prove that DGP achieves a stronger privacy protection in the sense of Definition~\ref{Def:attack}.
\begin{theorem}
\label{theorem1}
For any $(\varepsilon, \delta)$-passive attack  $\mathcal{A}$, under the presence of \textnormal{\textnormal{DGP}}, it will be degenerated to $(\varepsilon+\sqrt{\gamma_1} ||\nabla \mathbf{W}||_2,\delta)$-passive attack if $\mathcal{D}_\mathcal{A}$ is measured by Euclidean distance,  and degenerated to $(\varepsilon +(1-\varepsilon)\sqrt{\gamma_1},\delta)$-passive attack if $\mathcal{D}_\mathcal{A}$ is measured by cosine distance.
\end{theorem}
Theorem~\ref{theorem1} is based on Assumption~\ref{Assumption:A3} about DGP. It reveals that, with the same successful chance $(1-\delta)$, DGP weakens the passive attack  $\mathcal{A}$'s capability to obtain a better estimation of the true  $\nabla \mathbf{W}$. In particular, $\mathcal{A}$'s estimation of $\nabla \mathbf{W}$ is enlarged by $\sqrt{\gamma_1} ||\nabla \mathbf{W}||_2$ under Euclidean distance and enlarged by $(1-\varepsilon)\sqrt{\gamma_1}$ under cosine distance.

\subsection{Convergence Guarantee}
\label{Convergence}
 We start the convergence analysis by proving the generalization of DGP. The generalization analysis aims to quantify how the trained model performs on the test data, and it is achieved by analyzing the how DGP affects the properties of the optima reached (without gradient pruning)~\cite{EF,marginal}. 
 For ease of expression, let CL-SGD represent the training in CL with the SGD optimizer. 
 Based on Assumptions~\ref{Assumption:A3} and \ref{Assumption:A2}, the following Lemma can be obtained.


\begin{lemma}
\label{lemma:L_EF}
Let $\mathbf{e}_t={\sum_{i=1}^{N} \mathbf{e}_{t,i}}/{N}$ be the averaged accumulated error among all users at iteration $t$, the expectation of the norm of $\mathbf{e}_t$ is bounded, i.e.,
\begin{equation}
\mathbb{E}|| \mathbf{e}_{t} ||^2_2\le\frac{3\gamma_1(2+\gamma_1)}{2(1-\gamma_1)^2} G^2.
\end{equation}
\end{lemma}
Note that the difference between the averaged pruned gradient $\mathbf{g}_t={\sum_{i=1}^{N} \mathbf{g}_{t,i}}/{N}$ and the averaged {collaborative} SGD gradient $\mathbf{\nabla \mathbf{W }}_t={\sum_{i=1}^{N} \nabla \mathbf{W}_{t,i}}/{N}$ is simply $||\sum_{i=0}^{T-1}(\nabla \mathbf{W}_t-\mathbf{g}_t)||^2_2 =||\mathbf{e}_{T}||^2_2$. So the lemma above indicates that the accumulated gradient difference between our algorithm and {CL-SGD} is bounded.~That said, the optima reached by DGP and the optima reached by {CL-SGD} will eventually be very close if the algorithm converge. 
Armed with Lemma~\ref{lemma:L_EF} and based on Assumptions~\ref{Assumption:A3},~\ref{Assumption:A1} and~\ref{Assumption:A2}, we demonstrate the convergence of the our algorithm.

\begin{theorem}
The averaged norm of the full gradient $\nabla l(\mathbf{W }_t)$ derived from centralized training is correlated with the our algorithm as follows: 
\begin{IEEEeqnarray}{rCl}
\frac{\sum_{t=0}^{T-1} \mathbb{E}||\nabla l(\mathbf{W }_t)||_2^2}{T} &\le&4\frac{l^0-l^*}{\eta T} +2K\eta(G^2+\sigma^2)\nonumber
\\&+& 4\eta^2K^2\frac{3\gamma_1(2+\gamma_1)}{2(1-\gamma_1)^2}G^2,
\end{IEEEeqnarray}
where $l^0$ is the initialization of $l$, and $\eta$ is the learning rate.
\label{theorem:T2}
\end{theorem}
{The implication of Theorem~\ref{theorem:T2} is that, with an appropriate learning rate $\eta$, DGP converges similar to {CL-SGD} (slower by a negligible term $\mathcal{O}( \frac{1}{\sqrt{T} })$), as shown in} Corollary~\ref{C2}.
\begin{corollary}
\label{C2}
Let $\eta $=${({l^0-l^*})/{KT(G^2+\sigma^2)} }$, we have 
\begin{IEEEeqnarray}{rCl}
\frac{\sum_{t=0}^{T-1} \mathbb{E}||\nabla l(\mathbf{W}_t)||_2^2}{T} &\le& 6\sqrt{\frac{K(l^0-l^*)(\sigma^2+G^2 )}{T} }\nonumber\\&+&\mathcal{O}(\frac{1}{T})\nonumber.
\end{IEEEeqnarray}

\end{corollary}

\section{Experiments}\label{experiments}
\subsection{Experimental Setup}
We run the experiments with PyTorch by using one RTX 2080 Ti GPU and a 2.10 GHz CPU.
For fair comparison, we follow the setting of \cite{ATS}, using ten users with the same data distribution. {
We assess model privacy against various attacks and evaluate model performance on CIFAR10 and CIFAR100, which is a common setting used in many studies~\cite{evaluating,ATS}.}
We follow~\cite{evaluating,gradient} to quantify the privacy effect of defenses, i.e., visualizing the reconstructed data and using learned perceptual image patch similarity~(LPIPS) and structural similarity~(SSIM) to measure the quality of the recovered data.
A better defense should have larger LPIPS ($\uparrow$) and  smaller SSIM ($\downarrow$).

\noindent\textbf{Attack  methods.} 
We evaluate DGP against IG, GI, R-gap, and Rob attacks, which represent state-of-the-art passive and active GIAs, as discussed in Sec.~\ref{Sec:ThreatModel}.
{We use the following default attack settings: ResNet18 for IG , GI, Rob on CIFAR10. And we apply R-gap with CNN6~\cite{R-GAP} on CIFAR10, as this analysis attack is only suitable for models with simple structures. We provide additional attack details, more privacy evaluations (\eg more models and datasets)  and efficiency evaluation (computation costs and communication costs) in the appendix.}

\noindent\textbf{Defense  methods.} We compare DGP with six state-of-the-art defenses: Soteria, ATS, Precode, Outpost, DP and Top-$k$  pruning.
Besides, we set CL-SGD as the baseline that adopts no defense. 
{Note that DP provides privacy guarantee by adding noise to gradients in deep learning. We adhere to the DP settings of \cite{soteria} and use Gaussian noise with standard deviation $\sigma=10^{-2}$.}
When quantifying the defense performance of ATS, we not only evaluate the similarity between the raw images and the recovered data (ATS-T), but also evaluate the similarity between the disturbed training images (\ie, the real inputs) and the recovered data (ATS-R).
For Top-$k$ and DGP, we set $k=20\%$, $k_1+k_2=80\%$ with the regulation hyperparameter $p=1/15$. 
The rest defenses remain the original settings.
\subsection{Privacy Evaluation}
\label{sec:privacy}
\begin{table*}[htbp]
  \centering
     {
    \fontsize{9pt}{1pt}\selectfont
  \setlength{\tabcolsep}{1pt}  
 \begin{tabular}{|c|c|r|c|c|r|r|r|r|r|r}
    \toprule
    \multicolumn{1}{l|}{Attack} & \multicolumn{1}{l|}{Metric} & \multicolumn{1}{l|}{Baseline} & \multicolumn{1}{l|}{ATS-R} & \multicolumn{1}{l|}{ATS-T} & \multicolumn{1}{l|}{Soteria} & \multicolumn{1}{l|}{Precode} & \multicolumn{1}{l|}{DP} & \multicolumn{1}{l|}{Top-$k$} & \multicolumn{1}{l|}{Outpost} & \multicolumn{1}{l}{DGP} \\
    \midrule
    \multicolumn{1}{c|}{\multirow{2}[2]{*}{R-gap}} & \multicolumn{1}{l|}{LPIPS} & 7.7E-4  & \multicolumn{1}{r|}{1.3E-4 } & \multicolumn{1}{r|}{0.020 } & 0.378  &   -  & 0.373  & \textbf{0.379 } & \underline{0.378}  & 0.375  \\
    \multicolumn{1}{c|}{} & \multicolumn{1}{l|}{SSIM} & 0.965  & \multicolumn{1}{r|}{0.989 } & \multicolumn{1}{r|}{0.870 } & 0.252  &   -  & 0.259  & \underline{0.249}  & 0.250  & \textbf{0.248 } \\
    \midrule
    \multicolumn{1}{c|}{\multirow{2}[2]{*}{IG}} & \multicolumn{1}{l|}{LPIPS} & 0.003  & \multicolumn{1}{r|}{4.5E-4} & \multicolumn{1}{r|}{0.108 } & 0.190  & \textbf{0.371 } & 0.268  & 0.029  & 0.088  & \underline{0.316}  \\
    \multicolumn{1}{c|}{} & \multicolumn{1}{l|}{SSIM} & 0.954  & \multicolumn{1}{r|}{0.981 } & \multicolumn{1}{r|}{0.566 } & 0.368  & \textbf{0.257} & 0.333  & 0.769  & 0.640  & \underline{0.287}  \\
    \midrule
    \multicolumn{1}{c|}{\multirow{2}[2]{*}{GI}} & \multicolumn{1}{l|}{LPIPS} & 0.004  & \multicolumn{1}{r|}{0.003 } & \multicolumn{1}{r|}{0.094 } & 0.201  & \textbf{0.453 } & 0.343  & 0.045  & 0.111  & \underline{0.382}  \\
    \multicolumn{1}{c|}{} & \multicolumn{1}{l|}{SSIM} & 0.918  & \multicolumn{1}{r|}{0.908 } & \multicolumn{1}{r|}{0.563 } & 0.362  & \underline{0.247}  & 0.305  & 0.697  & 0.612  & \textbf{0.199} \\
    \midrule
    \multicolumn{1}{c|}{\multirow{4}[2]{*}{Rob}} & \multicolumn{1}{l|}{LPIPS} & 0.023  & \multicolumn{1}{r|}{0.028} & \multicolumn{1}{r|}{0.150 } & 0.023  & 0.025  & 0.023  & \underline{0.523}  & 0.295  & \textbf{0.527} \\
    \multicolumn{1}{c|}{} & \multicolumn{1}{l|}{Min LPIPS} & 7.43E-15 & \multicolumn{1}{r|}{5.03E-15} & \multicolumn{1}{r|}{0.011 } & 7.79E-15 & 5.52E-15 & 8.79E-07 & \underline{0.231}  & 0.195  & \textbf{0.243 } \\
    \multicolumn{1}{c|}{} & \multicolumn{1}{l|}{SSIM} & 0.933  & \multicolumn{1}{r|}{0.926 } & \multicolumn{1}{r|}{0.514} & 0.933  & 0.929  & 0.899  & \textbf{0.038} & 0.221  & \underline{0.051}  \\
    \multicolumn{1}{c|}{} & \multicolumn{1}{l|}{Max SSIM} & 1.000  & \multicolumn{1}{r|}{1.000 } & \multicolumn{1}{r|}{0.931 } & 1.000  & 1.000  & 1.000  & \textbf{0.224} & \underline{0.310}  & 0.365  \\
    \midrule
    \multicolumn{2}{c|}{Final Model Acc.} & \textbf{93.62\%} & \multicolumn{2}{c|}{93.14\%} &  92.90\%     & 92.83\%
 & 76.01\% & \underline{93.44\%} & 92.96\%
 & 93.40\% \\
    \bottomrule
    \end{tabular}%
    }
    \caption{
    Evaluation of the defense performance under four attacks.}
 \label{tab:four attacks}%
\end{table*}%

Tab.~\ref{tab:four attacks} shows the defense performance  with SSIM, and LPIPS  under four attacks. For each metric, we bold the best result and underline the second best result (the same hereinafter). The results show that ATS, Soteria, Precode, DP perform poorly under Rob attack, while Top-$k$ and Outpost are vulnerable to IG attack and GI attack.
{In summary, DGP can provide excellent privacy protection under all attacks, while still retain high model accuracy.}
To perceptually demonstrate the defense performance, we also visualize the reconstructed images. Note that ATS-T refers to processed raw data, while ATS-R represents the reconstructed raw data in Fig.~\ref{fig:privacy}.
Fig.~\ref{fig:IVG} and Fig.~\ref{fig:GI} depict the recovered images under optimization attacks (\eg,~IG, GI). We can find that the attacker can still recover the outline of inputs with ATS, Top-$k$ and Outpost. Soteria, Precode, DP and DGP can make the recovered images unrecognizable. 
Fig.~\ref{fig:RGAP} shows the recovered images from the R-gap attack.  We can see that all defenses but ATS can well defend against R-gap  because ATS does not damage the gradient structure, validating that a slight perturbation on gradients can mitigate the analytical attacks easily. 
We are not able to provide the result of Precode because its VB operation destroys the model structure, making R-gap cannot be mounted.
Fig.~\ref{fig:rob} plots the recovered images from the Rob attack. It shows that ATS, Precode, and Soteria fail to work and most inputs can be reconstructed. 
{
Fig.~\ref{fig:rob} shows that DP also cannot defend against Rob. This might be because the server calculates the inputs by superimposing a large number of the malicious imprint module's gradient parameters. And the noise added to the gradient follows a normal distribution, potentially canceling out when aggregated in large numbers. 
}
However, DGP, Top-$k$, and Outposts can effectively defend against Rob attack because the gradients of all layers are pruned , including those of the malicious imprint modules. 
However, we reiterate that the main weakness of the gradient pruning based on Top-$k$ selection is its vulnerability to optimization attacks (\eg, IG, GI), as widely demonstrated in the literature~\cite{ATS, soteria}.

\subsection{Accuracy Evaluation}
\label{sec:acc}

Tab.~\ref{tab:four attacks} lists the accuracy of ResNet18 on CIFAR10 under different defenses. 
Clearly, ATS, Soteria, Precode, Outpost, Top-$k$ and our method can achieve model accuracy similar to the unprotected baseline, while DP performs worst as expected. 
{Additionally, we evaluated more model performance with DGP, including ResNet18, VGG11~\cite{vgg13}, CNN6, LeNet~(Zhu)~\cite{IVG}. And we further perform ablation experiments to explore the role of the error feedback mechanism. Fig.~\ref{fig:diff_acc} shows that the model performance of DGP with error feedback is close to the baseline. However, DGP without error feedback performs poorly and even fails to converge. This is because accumulated errors result in a larger disparity between the model's update direction and the correct update direction. Notably, this effect is mitigated in structurally complex models due to the presence of numerous redundant parameters. Prior research \cite{wo_ef} indicated that even if these redundant parameters are not updated (\ie, their gradient parameters are set to 0), their impact on model performance is small.}
Our theoretical analysis and Fig.~\ref{fig:diff_acc} show that the error feedback mechanism can effectively correct the negative effects caused by gradient pruning. And Top-$k$ method can also enjoy the benefit since it is also based on pruning.  However, further experiments (see details in the appendix) validate that, to achieve a similar level of privacy protection of DGP with $80\%$ pruning, the pruning rate of Top-$k$ exceeds $95\%$ and results in inferior accuracy.
\begin{figure}[t!]
  \centering
  \subfigure[IG,~CIFAR10]{
   \label{fig:IVG}\includegraphics[width=0.45\textwidth]{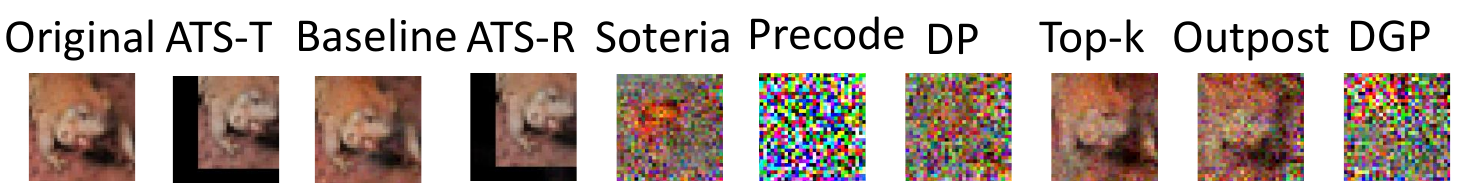}}
  \subfigure[GI,~CIFAR10]{
\label{fig:GI}\includegraphics[width=0.45\textwidth]{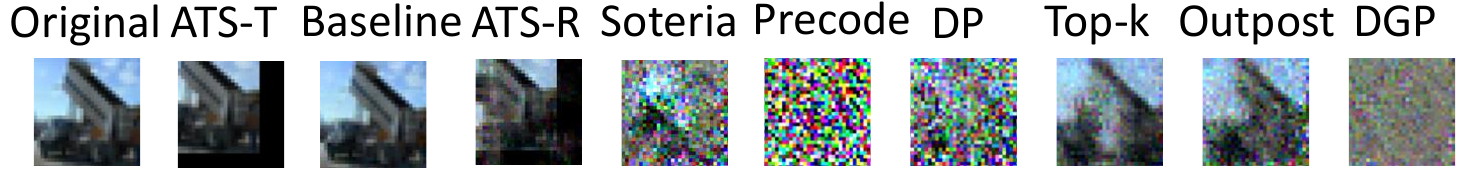}}
\subfigure[R-gap,~CIFAR10]{
\label{fig:RGAP}\includegraphics[width=0.43\textwidth]{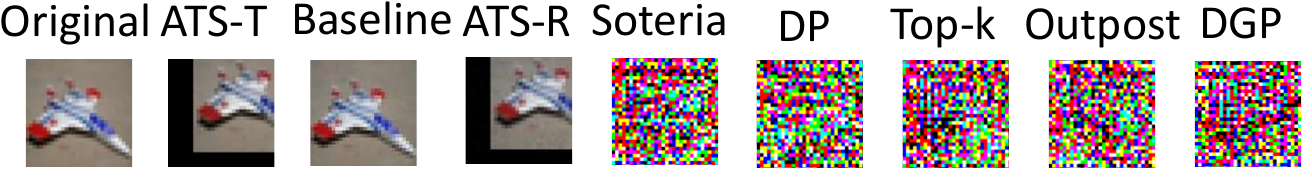}}
\subfigure[Rob,~CIFAR10]{
\label{fig:rob}\includegraphics[width=0.45\textwidth]{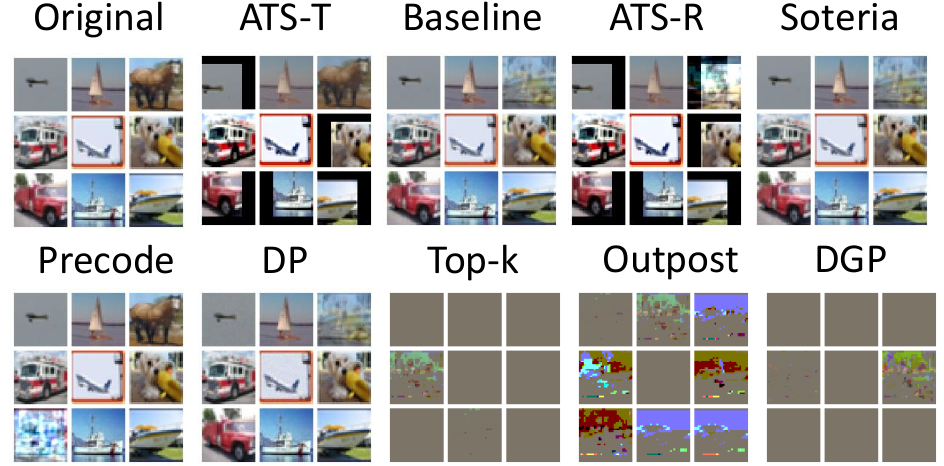}}
\caption{Data visualization on privacy evaluation by using multiple gradient inversion attacks.}
       \label{fig:privacy}
\end{figure}


\begin{figure}[t!]
    \centering
    \includegraphics[width=0.33\textwidth]{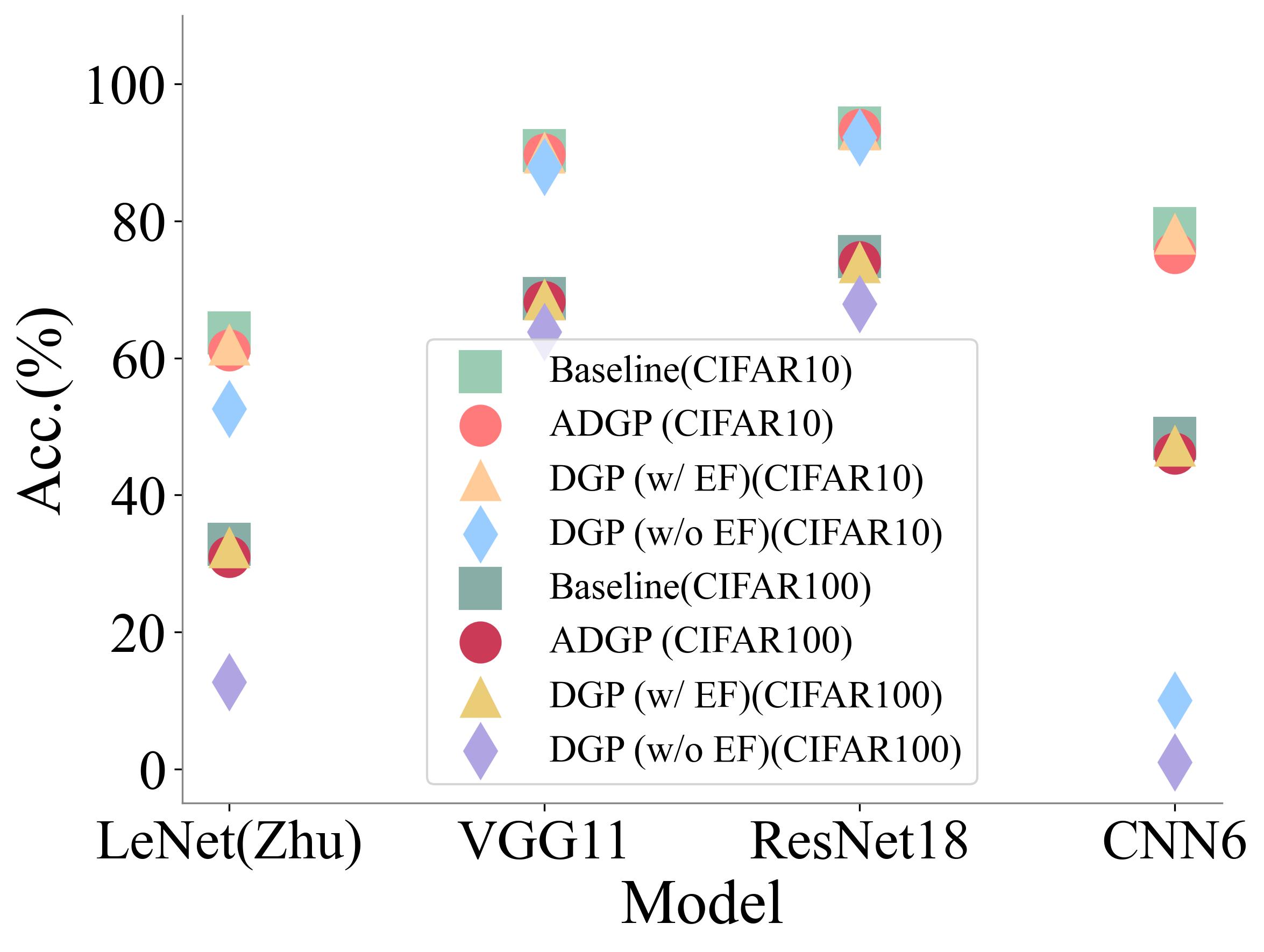}
        \centering
\caption{Evaluation of model accuracy with different datasets and models (EF denotes the error feedback).}
  \label{fig:diff_acc}
\end{figure}

\subsection{Further Discussions}
\label{sec:download}
\noindent\textbf{Choice of $k_1$, $k_2$ and $p$ for DGP.}~
According to the analysis in Sec.~\ref{Sec:rob_ivg}, active GIA is greatly impacted by $(k_1+k_2)$ and optimization GIA is greatly affected by $p=k_1/k_2$. 
In this concern, we use the Rob attack to evaluate the privacy of DGP with different $(k_1+k_2)$ and IG attack to evaluate DGP with different $p$.
As shown in Tab.~\ref{tab:different_p_k}, larger pruning rate $(k_1+k_2)$ leads to better privacy-preserving, but the model's performance suffers as a consequence. 
Furthermore, a larger $p$, i.e., more large parameters are eliminated, can better defend against optimization GIAs  but impact accuracy.



\begin{table}[t!]
  \centering
  { \fontsize{9pt}{3.5pt}\selectfont
  \setlength{\tabcolsep}{3pt} 
    \begin{tabular}{|l|rrr|rrr|}
    \toprule
    \multirow{2}[4]{*}{} & \multicolumn{3}{c|}{($k_1 + k_2$)} & \multicolumn{3}{c|}{$p=k_1/k_2$} \\
\cmidrule{2-7}          & 48\%  & 80\%  & 96\%  & 1/15   & 1/7 & 1/3 \\
    \midrule
    LPIPS & 0.426 & 0.527 & 0.531 & 0.316  & 0.351  & 0.383  \\
    SSIM  & 0.146 & 0.051 & 0.029 & 0.287  & 0.250  & 0.234  \\
    Acc.(\%) & 93.42  & 93.40  & 92.91  & 93.40  & 93.21  & 92.82  \\
    \bottomrule
    \end{tabular}%
    }
     \caption{The impact of different parameters on DGP.}
  \label{tab:different_p_k}%
\end{table}%


\noindent\textbf{Reducing download communication cost.}~Although DGP provides a sufficient privacy guarantee as well as reducing upload cost, users' download cost could still be expensive. 
This is because different users have different sets of $\mathcal{T}_{k_1}(\cdot)$ and $\mathcal{B}_{k_2}(\cdot)$ when pruning their own local gradients, so the global model parameters will become dense after aggregation.
We suggest aligned DGP (ADGP), an improved scheme to align the selected gradients to further reduce download cost.
Similar to DGP, for best privacy, each user will still firstly identify his top-$k_1$ gradients location set $\mathcal{T}_{k_1}$.
Different from DGP, ADGP also wants to save users' download comm. cost by ensuring that all users' uploaded pruned gradient parameters reside in the same location set. 
This is achieved by randomly selecting a user, who identifies a top-$2k$ ($k_1<k$) location set $\mathcal{T}_{2k}$ (represented with a binary location matrix $\mathcal{I}$) and broadcasts $\mathcal{I}$ to all other users. Note that $\mathcal{T}_{k_1}  \subset \mathcal{T}_{2k}$  is not necessarily true. Upon receiving $\mathcal{I}$, each user first discards gradient parameters in $\mathcal{T}_{k_1}$ and then only transmits the $k$ largest gradient parameters whose locations belong to $\mathcal{I}$. 
After aggregation, users only need to download the global gradients' parameters associated with $\mathcal{I}$. 
We give the specific comm. cost in the appendix and find that ADGP further reduces the overall comm. cost. 
Moreover, with error feedback mechanism, it can also maintain the model performance, shown in Fig.~\ref{fig:diff_acc}. To summarize, ADGP can provide better communication efficiency while maintain model performance. We leave the work of investigating the privacy-protection of ADGP as the future work. 

\section{Conclusion, Limitation, and Future}\label{conclusion}
Contrary to the traditional belief that gradient pruning is not a good choice to protect privacy, this paper proposes DGP, a gradient pruning-based defense, to achieve a better trade-off among privacy protection, model performance, and communication efficiency for collaborative learning.
This finding is built upon the analysis of how pruned gradients bound the attacker's recovery error and why large gradient parameters leak more private information and should be pruned. By dual-pruning both large and small gradients, DGP guarantees 
theoretical convergence and better privacy protection against passive attackers. 
By comparing to state-of-the-art defenses,  experimental results corroborate our theoretical analysis, as well as empirically demonstrating the advantage of DGP against active attackers. 
In terms of limitations, the success of ADGP relies on selecting a reliable user to broadcast its locations. When this user becomes malicious, the entire system will fail. 
In the future, we will provide more rigorous and more comprehensive privacy analysis, investigate the privacy property of ADGP under passive attacks, explore the applications of (A)DGP in federated learning and broaden our research to more domains like NLP.

\section*{Acknowledgements} 
Shengshan's work is supported in part by the National Natural Science Foundation of China (Grant No.U20A20177) and Hubei Province Key R\&D Technology Special Innovation Project under Grant No.2021BAA032.
Shengqing's work is supported in part by Hubei Provincial Natural Science Foundation Project (NO. 2023AFB342) and Open Program of Nuclear Medicine and Molecular Imaging Key Laboratory of Hubei Province (NO. 2022fzyx018). The work is supported by HPC Platform of Huazhong University of Science and Technology.
Shengshan Hu is the corresponding author.

\bibliography{aaai24}
\newpage
\title{Appendix for "Revisiting Gradient Pruning: A Dual Realization for Defending Against Gradient Attacks."}
\maketitle
\thispagestyle{empty}
\appendix




\maketitle
\section{Theoretical proof}
This section presents all the missing theoretical analyses appeared in the manuscript orderly. 

\begin{proposition}
\label{theo:errorbound}
For any given input $\mathbf{x}$ and shared model $\mathbf{W}$, the distance between the recovered data $\mathbf{x}'$ and the real data $\mathbf{x}$ is bounded by:
\begin{equation*}
||\mathbf{x} -\mathbf{x}'||_2  \ge \frac{||\varphi (\mathbf{x},\mathbf{W})- \varphi (\mathbf{x}',\mathbf{W})||_2}
{||\partial \varphi (\mathbf{x},\mathbf{W})/\partial\mathbf{x}||_2},
\end{equation*}
\end{proposition}
\begin{proof}
Apply the first-order Taylor expansion to $(\varphi (\mathbf{x},\mathbf{W})- \varphi (\mathbf{x'},\mathbf{W}))$, it is easy to find
\begin{IEEEeqnarray}{rCL}
     &||&\varphi (\textbf{x},\textbf{W})-\varphi (\textbf{x}',\textbf{W})||_2\nonumber\\
     &\approx& ||(\partial \varphi (\mathbf{x},\mathbf{W})/\partial\mathbf{x})(\textbf{x}-\textbf{x}')||_2 \nonumber \\
     &\le& ||(\partial \varphi (\mathbf{x},\mathbf{W})/\partial\mathbf{x})||_2||(\textbf{x}-\textbf{x}')||_2. \nonumber
\end{IEEEeqnarray}
Hence, we have
\begin{align}
&||\mathbf{x} -\mathbf{x}'||_2  \ge \frac{||\varphi (\mathbf{x},\mathbf{W})- \varphi (\mathbf{x}',\mathbf{W})||_2}
{||\partial \varphi (\mathbf{x},\mathbf{W})/\partial\mathbf{x}||_2}.
\end{align}
\end{proof}

\begin{theorem}
\label{theorem1}
For any $(\varepsilon, \delta)$-passive attack  $\mathcal{A}$, under the presence of \textnormal{\textnormal{DGP}}, it will be degenerated to $(\varepsilon+\sqrt{\gamma_1} ||\nabla \mathbf{W}||_2,\delta)$-passive attack if $\mathcal{D}_\mathcal{A}$ is measured by Euclidean distance,  and degenerated to $(\varepsilon +(1-\varepsilon)\sqrt{\gamma_1},\delta)$-passive attack if $\mathcal{D}_\mathcal{A}$ is measured by cosine distance.
\end{theorem}
\begin{proof}
If $\mathcal{D}_\mathcal{A}$ is measured by Euclidean distance, by the definition of $(\varepsilon,\delta)$-attack, the attacker can achieve the following estimation 
\begin{align*}
&\mathbb{E}|| \mathbf{g}^*-\nabla \textbf{W}||_2\le\varepsilon,
\end{align*}
where $ \mathbf{g}^*$ is the attacker's optimized gradients of the ground-truth gradients $\textbf{W}$.
When \textnormal{\textnormal{DGP}} is used, from the bi-Lipschitz assumption (i.e., Assumption 1), we know
\begin{align}
    &  || \textnormal{\textnormal{DGP}}(\nabla \textbf{W})-\nabla \textbf{W}||_2 
    {\le} \sqrt{\gamma_1 } ||\nabla \textbf{W}||_2.
\label{Eq:fromassumption1}
\end{align}
%
Then, when central aggregation is protected by \textnormal{DGP}, the attacker's optimized gradients is based on the observation of $\textnormal{\textnormal{DGP}}(\nabla \textbf{W})$ and this modified observation will degrade the attacker's capability in optimizing $\nabla \textbf{W}$ because
\begin{IEEEeqnarray}{rCL}
 &\mathbb{E}&|| \mathbf{g}^*-\nabla \textbf{W}||_2 \nonumber  
\\ & = &\mathbb{E}|| \mathbf{g}^*-\textnormal{\textnormal{DGP}}(\nabla \textbf{W})+\textnormal{\textnormal{DGP}}(\nabla \textbf{W})-\nabla \textbf{W}||_2 \nonumber \\
 &{\le }& \varepsilon +||\textnormal{\textnormal{DGP}}(\nabla \textbf{W})-\nabla \textbf{W}||_2 \nonumber \\
 & \le&  \varepsilon +\sqrt{\gamma_1 } ||\nabla \textbf{W}||_2. \nonumber
\end{IEEEeqnarray}
Hence, the first part of this theorem is true.

Similarly, when $\mathcal{D}_\mathcal{A}$ is measured by cosine distance, the definition of $(\varepsilon,\delta)$-attack reveals
\begin{align*}
    \mathbb{E}\left[ 1- \frac{< \mathbf{g}^*, \nabla \textbf{W} >} 
    {|| \mathbf{g}^*||_2||, \nabla \textbf{W}||_2} \right]
    \le \varepsilon.
\end{align*}
{Then, we can obtain}

\begin{IEEEeqnarray}{rCL}
&\mathbb{E}&\left [ \frac{< \mathbf{g}^*,\nabla \textbf{W}> }{|| \mathbf{g}^*||_2||\nabla \textbf{W}||_2} \right ] \nonumber \\
&=&\mathbb{E}\left [ \frac{< \mathbf{g}^*,\nabla \textbf{W}-\textnormal{DGP}(\nabla \textbf{W})+\textnormal{DGP}(\nabla \textbf{W})> }{|| \mathbf{g}^*||_2||\nabla \textbf{W}||_2} \right ]\nonumber \\
&\overset{(c)}{=}& \mathbb{E}\left [ \frac{< \mathbf{g}^*,\textnormal{DGP}(\nabla \textbf{W})> }{|| \mathbf{g}^*||_2||\nabla \textbf{W}||_2} \right ] \nonumber \\
&=& \mathbb{E}\left [ \frac{< \mathbf{g}^*,\textnormal{DGP}(\nabla \textbf{W})> }{|| \mathbf{g}^*||_2||\textnormal{DGP}(\nabla \textbf{W})||_2}\frac{||\textnormal{DGP}(\nabla \textbf{W})||_2}{||\nabla \textbf{W}||_2}  \right ] \nonumber \\
&\overset{(d)}{\ge}&  (1-\sqrt{\gamma_1})\mathbb{E}\left [ \frac{< \mathbf{g}^*,\textnormal{DGP}(\nabla \textbf{W})> }{|| \mathbf{g}^*||_2||\textnormal{DGP}(\nabla \textbf{W})||_2}\right ] \nonumber \\
&\ge& (1-\sqrt{\gamma_1})(1-\varepsilon)=1+\varepsilon\sqrt{\gamma_1}-\sqrt{\gamma_1}-\varepsilon.
\label{Eq:cosineEstimation}
\end{IEEEeqnarray}\\
where (c) is based on the fact that the all non-zero elements of $(\nabla \textbf{W}-\textnormal{DGP}(\nabla \textbf{W}))$ are pruned in DGP so $\mathbb{E}( \mathbf{g}^*, (\nabla \textbf{W}-\textnormal{DGP}(\nabla \textbf{W}))) = 0$, and (d) is the direct application of Eq.~(\ref{Eq:fromassumption1}). 
Based on Eq.~(\ref{Eq:cosineEstimation}), it is easy to conclude
\begin{IEEEeqnarray}{rCL}
&\mathbb{E}\left [1- \frac{< \mathbf{g}^*,\nabla \textbf{W}> }{|| \mathbf{g}^*||_2||\nabla \textbf{W}||_2} \right ] \le \varepsilon +(1-\varepsilon)\sqrt{\gamma_1},\nonumber
\end{IEEEeqnarray}
which completes the proof. 
\end{proof}

\begin{lemma}
\label{lemma:L_EF}
Let $\mathbf{e}_t={\sum_{i=1}^{N} \mathbf{e}_{t,i}}/{N}$ be the averaged accumulated error among all users at iteration $t$, the expectation of the norm of $\mathbf{e}_t$ is bounded, i.e.,
\begin{equation*}
\mathbb{E}|| \mathbf{e}_{t} ||^2_2\le  \frac{3\gamma_1(2+\gamma_1)}{2(1-\gamma_1)^2}G^2.
\end{equation*}
\end{lemma}
\begin{proof}
To use the theoretical tools of SGD, we set up the following dummy matrix $\textbf{V}$: 
\begin{equation*}
\textbf{V}_{t+1}=\textbf{V}_t-\eta \nabla \textbf{W}_t.
\end{equation*}
Since $\textbf{W}^0=\textbf{V}^0$, $\textbf{e}^0=0$, it is easy to find
\begin{equation}
\textbf{V}_{t}-\textbf{W}_t=\eta \textbf{e}_t.
\label{equation:e}
\end{equation}
Under Assumption~1, we have
\begin{align}
\label{equation:A31}
&||\textbf{X}-\textnormal{DGP}(\textbf{X})||_2^2 \le \gamma_1 ||\textbf{X}||_2^2
\end{align}
Under Assumption~3, we have
\begin{align}
\label{equation:A21}
\mathbb{E}||\nabla \textbf{W}_{t,i}||_2^2 \le G^2
\end{align}
\begin{align}
\label{equation:A22}
\mathbb{E}||\nabla \textbf{W}_t||_2^2\le G^2+\sigma^2.
\end{align}
By definition of $\textbf{e}_{t}$, we know 
\begin{IEEEeqnarray}{rCL}
||\textbf{e}_{t}||_2^2 &\le& \frac{\sum_{i=1}^{N} ||\textbf{e}_{t,i}||_2^2}{N}, \nonumber  
\end{IEEEeqnarray}
and the $||\textbf{e}_{t,i}||_2^2$ is also bounded because
\begin{IEEEeqnarray}{rCL}
||\textbf{e}_{t,i}||_2^2 
&=&||\nabla \textbf{W}_{t-1,i}+\textbf{e}_{t-1,i}-\textnormal{DGP}(\nabla \textbf{W}_{t-1,i}+\textbf{e}_{t-1,i})||_2^2\nonumber\\
&\overset{(\ref{equation:A31})}{\le}&  \gamma_1 ||\nabla \textbf{W}_{t-1,i}+\textbf{e}_{t-1,i}||_2^2 \nonumber\\ 
&\overset{(e)}{\le}& \gamma_1\left( (1+\frac{1}{a} )||\nabla \textbf{W}_{t-1,i}||_2^2+(1+a)||\textbf{e}_{t-1,i}||_2^2 \right). \nonumber
\end{IEEEeqnarray}
where (e) is based on the variant of Young’s inequality  $||x+y||_2^2\le(1+a)||x||_2^2+(1+\frac{1}{a} )||y||_2^2$. 
Set $1+a=\frac{2+\gamma_1 }{3\gamma_1 }$, it is concluded that
\begin{align}
&\mathbb{E}||\textbf{e}_{t,i}||_2^2 
\overset{(\ref{equation:A21})}{\le } \frac{3\gamma_1(2+\gamma_1)}{2(1-\gamma_1)^2}G^2, \\ 
\label{Inequality:EF}
&\mathbb{E}||\textbf{e}_{t}||_2^2\le  \frac{3\gamma_1(2+\gamma_1)}{2(1-\gamma_1)^2}G^2.
\end{align}
\end{proof}    
\begin{theorem}
The averaged norm of the full gradient $\nabla l(\mathbf{W }_t)$ derived from centralized training is correlated with the our algorithm as follows: 
\begin{align}
\frac{\sum_{t=0}^{T-1} \mathbb{E}||\nabla l(\mathbf{W }_t)||_2^2}{T} &\le 
4\frac{K^0-l^*}{\eta T}+4\eta^2K^2\frac{3\gamma_1(2+\gamma_1)}{2(1-\gamma_1)^2}G^2\nonumber\\&+2K\eta(G^2+\sigma ^2),
\end{align}
where $l^0$ is the initialization of the objective $l$, and $\eta$ is the learning rate.
\label{theorem:T2}
\end{theorem}

\begin{proof}
Under Assumption~3, we have
\begin{align}
\label{Inequality:I2}
    &||\nabla l(\textbf{V}_t) -\nabla l(\textbf{W}_t)||\le K||\textbf{V}_t-\textbf{W}_t||,
\end{align}
and
\begin{IEEEeqnarray}{rCL}
l(\textbf{V}_{t+1}) &\le& l(\textbf{V}_t)+<\nabla l(\textbf{V}_t),\textbf{V}_{t+1}-\textbf{V}_t>
+\frac{K}{2}||\textbf{V}_{t+1}-\textbf{V}_t||_2^2
\nonumber\\ 
&= & l(\textbf{V}_t)-\eta <\nabla l(\textbf{V}_t),\nabla \textbf{W}_t>+\frac{K\eta ^2}{2}||\nabla \textbf{W}_t||^2_2.
\label{Inequality:I1}
\end{IEEEeqnarray}
Taking expectation on both sides of Eq.~(\ref{Inequality:I1}), we can get
\begin{IEEEeqnarray}{rCL}
&\mathbb{E}&(l(\textbf{V}_{t+1}))
\le \mathbb{E}(l(\textbf{V}_t))-\eta \mathbb{E}{(<\nabla l(\textbf{V}_t), \nabla l(\textbf{W}_t)>)} \nonumber\\&+&\frac{K\eta ^2}{2}\mathbb{E}||\nabla \textbf{W}_t||^2_2  \nonumber \\
&\overset{(f)}{=}& \mathbb{E}(l(\textbf{V}_t))-\frac{\eta}{2} ( \mathbb{E}(||\nabla l(\textbf{V}_t)||_2^2+||\nabla l(\textbf{W}_t)||_2^2) \nonumber\\&+& \frac{\eta}{2}\mathbb{E}||\nabla l(\textbf{V}_t)-\nabla l(\textbf{W}_t)||_2^2
 +\frac{K\eta ^2}{2}\mathbb{E}||\nabla \textbf{\textbf{W}}_t||^2_2  \nonumber \\
&\le& \mathbb{E}(l(\textbf{V}_t))-\frac{\eta}{2}\mathbb{E}(||\nabla l(\textbf{V}_t)||_2^2)  \nonumber+\frac{K\eta ^2}{2}\mathbb{E}||\nabla \textbf{\textbf{W}}_t||^2_2 \\&+&\frac{\eta}{2}\mathbb{E}||\nabla l(\textbf{V}_t)-\nabla l(\textbf{W}_t)||_2^2  \nonumber \\
&\overset{(\ref{Inequality:I2})}{\le}& \mathbb{E}(l(\textbf{V}_t))-\frac{\eta}{2}(\mathbb{E}||\nabla l(\textbf{V}_t)||_2^2)\nonumber+\frac{K\eta ^2}{2}\mathbb{E}||\nabla {\textbf{W}}_t||^2_2\\&+&\frac{K^{2}\eta}{2}\mathbb{E}||\textbf{V}_t-\textbf{W}_t||_2^2 \nonumber \\
&\overset{(\ref{equation:e})}{\le}& \mathbb{E}(l(\textbf{V}_t))-\frac{\eta}{2}(\mathbb{E}||\nabla l(\textbf{V}_t)||_2^2)  +\frac{\eta^3 K^2}{2}\mathbb{E}||\textbf{e}_{t}||_2^2\nonumber\\&+&\frac{K\eta ^2}{2}\mathbb{E}||\nabla {\textbf{W}}_t||^2_2  \nonumber \\
&\overset{(\ref{equation:A22})}{\le}& \mathbb{E}(l(\textbf{V}_t))-\frac{\eta}{2}(\mathbb{E}||\nabla l(\textbf{V}_t)||_2^2)  \nonumber+\frac{\eta^3 K^2}{2}\mathbb{E}||\textbf{e}_{t}||_2^2\nonumber\\&+&\frac{K\eta ^2}{2}(G^2+{\sigma ^2 }),
\end{IEEEeqnarray}
{where (f) is based on the fact $<x, y> = \frac{1}{2}(||x||^2 + ||y||^2 -||x-y||^2)$.} Base on the deduction above, we can further calculate 
\begin{IEEEeqnarray}{rCL}
\frac{\eta}{2}(\mathbb{E}||\nabla l(\textbf{V}_t)||_2^2) 
&\le& \mathbb{E}(l(\textbf{V}_t))- \mathbb{E}(l(\textbf{V}_{t+1}))  +\frac{\eta^3 K^2}{2}\mathbb{E}||\textbf{e}_{t}||_2^2\nonumber\\&+&\frac{K\eta ^2}{2}(G^2+\sigma ^2 ),
\end{IEEEeqnarray}
\begin{IEEEeqnarray}{rCL}
(\frac{\sum_{0}^{T-1} \mathbb{E}||\nabla l(\textbf{V}_t)||_2^2}{T} ) 
&\le& \frac{ 2(l^0- l^*)}{\eta T}   +\eta^2 K^2\mathbb{E}||\textbf{e}_{t}||_2^2\nonumber\\&+&K\eta(G^2+\sigma ^2 )\label{E_V}.
\end{IEEEeqnarray}
According to~Eq.~(\ref{Inequality:I2}), it can be found that
\begin{align}
&||\nabla l(\textbf{W}_t)||\le K||\textbf{V}_t-\textbf{W}_t||+||\nabla l(\textbf{V}_t) ||,\nonumber\\
&||\nabla l(\textbf{W}_t)||_2^2\le 2 K^2||\textbf{V}_t-\textbf{W}_t||_2^2+2||\nabla l(\textbf{V}_t) ||_2^2\label{V_W}.
\end{align}
Combining~Eq.~(\ref{Inequality:EF}),~Eq.~(\ref{E_V}) and ~Eq.~(\ref{V_W}), it is concluded
\begin{IEEEeqnarray}{rCL}
&\mathbb{E}&||\nabla l(\textbf{W}_t)||_2^2
\le \frac{ 4(l^0- l^*)}{\eta T}  +4\eta^2 K^2\mathbb{E}||\textbf{\textbf{e}}_{t}||_2^2\nonumber\\&+&2K\eta(G^2+\sigma ^2) \nonumber \\
&\le& \frac{ 4(l^0- l^*)}{\eta T}+4\eta^2K^2\frac{3\gamma_1(2+\gamma_1)}{2(1-\gamma_1)^2}G^2 +2K\eta(G^2+\sigma ^2 ). \nonumber
\end{IEEEeqnarray}
Set $\eta=\sqrt{\frac{l^0-l^*}{KT(\sigma ^2+G^2 )} }$, we have
\begin{align}
&\frac{\sum_{0}^{T-1} \mathbb{E}||\nabla l(\textbf{W}_t)||_2^2}{T} 
\le 6\sqrt{\frac{K(l^0-l^*)(\sigma ^2+G^2 )}{T} }\ +\mathcal{O}(\frac{1}{T}). \nonumber
\end{align}
Hence, the theorem is true. 
\end{proof} 
\section{Analysis of Assumption}
Since Assumption 2 and Assumption 3 are common assumptions in many works~\cite{EF, marginal}, this section focus on analyzing the feasibility of Assumption 1.

\begin{assumption}
The pruning mechanism $\textnormal{DGP}(k_1, k_2, \cdot)$ is Lipschitz, so the following conditions  hold:
\begin{IEEEeqnarray}{rCl}
&||&\nabla \mathbf{W} -\textnormal{DGP}(k_1, k_2, \nabla \mathbf{W})  ||_2^2 
\nonumber\\
&=& ||\textnormal{DGP}(0, 0, \nabla \mathbf{W})-\textnormal{DGP}(k_1, k_2, \nabla \mathbf{W})  ||_2^2  \le \gamma_1 ||\nabla \mathbf{W}||_2^2,\nonumber
\end{IEEEeqnarray}
where $\gamma_1$ is a constant related to $k_1$ and $k_2$ and satisfies $(1-\sqrt{1-k_1*k_2} )^2<\gamma_1<1$.
\end{assumption}

To simplify the expression, we use $\textnormal{DGP} (\nabla \mathbf{W})$ to denote $\textnormal{DGP} (\nabla \mathbf{W},k_1,k_2)$ and $||\cdot||$ to denote $||\cdot||_2$.~\cite{Top-k} states the following property of top$[l](\nabla \mathbf{W})$ (i.e., retain the top $l$-ratio of $\nabla \mathbf{W}$):
\begin{equation}
||\nabla \mathbf{W} - \textnormal{top}[l](\nabla \mathbf{W})|| \leq \sqrt{1-l} ||\nabla \mathbf{W}||.
\label{topk:equ0}
\end{equation}
According to formula \ref{topk:equ0}, it is easy to obtain formula \ref{topk:equ01}:
\begin{equation}
||\textnormal{top}[l](|\nabla \mathbf{W})|| \ge(1-\sqrt{1-l} )||\nabla \mathbf{W}||.
\label{topk:equ01}
\end{equation}

And easy to find:
\begin{IEEEeqnarray}{rCl}
    &||&\nabla \mathbf{W}-\textnormal{DGP}(k_1,k_2,\nabla \mathbf{W})||_2\nonumber\\&=&||top[k_1](\nabla \mathbf{W})+bottom[k_2](\nabla \mathbf{W})||_2 \nonumber\\&> & ||top[k_1*k_2](\nabla \mathbf{W})||_2\nonumber.
\end{IEEEeqnarray}
Combined with formula \ref{topk:equ01}, it is true that $\gamma_1 > (1-\sqrt{1-k_1*k_2} )^2$.
Moreover, even if the difference is large, i.e., almost all parameters are removed and $\gamma_1$  approaches 1, the assumption still holds.

\section{More experimental results}
This section presents the related experimental setup and more experimental results.

\subsection{Experimental setup}
\textbf{Privacy experiment setup.}~~  About IG and GI, follow \cite{ATS}, we set the number of iterations to 2500, and we set the optimal learning rate to 0.1. For rob attack, we set the attack batchsize=9.
For IG, GI and R-GAP attacks, we set attack bachsize=1, randomly select 20 data points for attack, and calculate the average metrics of the attack results. 
In addition, there are some settings about defenses. We set the pruning rate of Soteria to 80, set the Outpost hyperparameters as $\lambda$= 0.8, $\varphi$=40, $\beta$= 0.1, $\rho$= 80, which are the originally experimental setting. 
In addition, we configured an LPIPS object employing AlexNet as the perceptual model, considering the spatial structure of images for comparison.


\textbf{Accuracy experiment setup.}~~ 
We train models with batchsize=32. We use SGD optimizer with momentum of 0.9 and set epoch=100. To ensure a good performance of the baseline, we set the following learning rates. For LeNet~(Zhu),  we set the learning rate $\eta$=0.1 if epoch $\leq $ 50, $\eta$=0.01 if epoch $>$50, and $\eta$=0.005 if epoch $>$70. For the rest of the training settings, we set the learning rate $\eta$=0.01 if epoch $\leq$ 70, $\eta$=0.005 if epoch $>$70. Considering the error feedback is designed for improving model accuracy for gradient pruning~\cite{EF}, we also apply error feedback to Top-$k$. 

\subsection{Efficiency Evaluation}
\textbf{Communication cost.} We measure the communication cost of the defenses for all trained models and list the average result of one epoch in Tab.~\ref{tab:comm}. DP and ATS do not affect the number of transmitted parameters, their results are the same as the baseline. Clearly, DGP and Top-$k$ save about $40\%$ bandwidth when comparing to the baseline. Outpost perturbs the gradient parameters based on Top-k pruning, so the result is consistent with Top-$k$.
Note that this advantage is free since gradient pruning incurs a negligible computation burden (detailed evaluation is presented in the Tab.~\ref{tab:comp}).
\begin{table}[htbp]
  \centering

  \scalebox{0.73}{
     \begin{tabular}{lrrrrrr}
    \toprule
           Model & \multicolumn{1}{l}{Baseline} & \multicolumn{1}{l}{Soteria} & \multicolumn{1}{l}{Precode} & \multicolumn{1}{l}{Top-$k$} & \multicolumn{1}{l}{DGP} & \multicolumn{1}{l}{ADGP} \\
    \midrule
    LeNet~(Zhu) & 0.121  & 0.098  & 4.631  & 0.074  & 0.074  & \textbf{0.038}  \\
    VGG11 & 70.428  & 70.413  & 73.436  & 43.137  & 43.137  & \textbf{22.449}  \\
    ResNet18 & 85.251  & 85.235  & 88.258  & 52.216  & 52.216  & \textbf{27.174}  \\
    CNN6  & 1.177  & 0.959  & 19.955  & 0.721  & 0.721  & \textbf{0.375}  \\
    \bottomrule
    \end{tabular}%
    }
      \caption{Average overall comm. cost in one epoch (MB).}
  \label{tab:comm}%

\end{table}%

\textbf{Computation cost.} Tab.~\ref{tab:comp} shows the computation cost comparison of gradient parameter searching for one epoch.
Although the average computation cost of ADGP is slightly higher than DGP because ADGP needs to load the binary matrix $\mathcal{I} $, this computation cost is trivial considering the reduced communication cost. And our method is obviously better than Soteria, because Soteria requires a lot of computation on gradients, which leads to expensive computation cost.
\begin{table}[htbp]
    \vspace{-0.5em}
  \centering

 \scalebox{0.6}{
        \begin{tabular}{lrrrr}
    \toprule
          & \multicolumn{1}{l}{Soteria} & \multicolumn{1}{l}{Top-$k$} & \multicolumn{1}{l}{DGP} & \multicolumn{1}{l}{ADGP} \\
    \midrule
    LeNet~(Zhu) & 52.460  & 0.113  & 0.146  & 0.167  \\
    VGG11 & 331.379  & 0.419  & 0.764  & 0.842  \\
    ResNet18 & 862.567  & 0.967  & 1.803  & 2.587  \\
    CNN6  & 291.419  & 0.100  & 0.237  & 0.256  \\
    \bottomrule
    \end{tabular}%
    }
    \caption{Average comp. cost in one epoch (s)}
        \vspace{-1em}
  \label{tab:comp}%
\end{table}%

\subsection{Privacy evaluation with more models.}
In this section, We evaluate privacy with multiple models and datasets. We choose the state-of-the-art active attack Rob and the state-of-the-art passive attack IG for privacy evaluation. As shown in Tab. \ref{tab:privacy_cifar10} and Tab. \ref{tab:privacy_cifar100}, DGP can play an effective privacy protection.

\begin{table*}[htbp]
  \centering
    \scalebox{0.8}{
    \begin{tabular}{|c|l|rr|rr|rr|rr|}
    \toprule
          &       & \multicolumn{2}{c|}{CNN6} & \multicolumn{2}{c|}{LeNet (Zhu)} & \multicolumn{2}{c|}{ResNet18} & \multicolumn{2}{c|}{VGG11} \\
    \midrule
    \multicolumn{1}{|l|}{Attack} & Metric  & \multicolumn{1}{l}{Baseline} & \multicolumn{1}{l|}{DGP} & \multicolumn{1}{l}{Baseline} & \multicolumn{1}{l|}{DGP} & \multicolumn{1}{l}{Baseline} & \multicolumn{1}{l|}{DGP} & \multicolumn{1}{l}{Baseline} & \multicolumn{1}{l|}{DGP} \\
    \midrule
    \multirow{2}[2]{*}{Rob attack} & LPIPS & 0.0216  & \textbf{0.4421 } & 0.0226  & \textbf{0.3031 } & 0.0231  & \textbf{0.5273 } & 0.0272  & \textbf{0.5079 } \\
          & SSIM  & 0.9273  & \textbf{0.1335 } & 0.9293  & \textbf{0.2364 } & 0.9328  & \textbf{0.0511 } & 0.9328  & \textbf{0.0464 } \\
    \midrule
    \multirow{2}[2]{*}{IG attack} & LPIPS & 0.0308  & \textbf{0.2349 } & 0.0788  & \textbf{0.2537 } & 0.0028  & \textbf{0.3163 } & 0.0303  & \textbf{0.2845 } \\
          & SSIM  & 0.7735  & \textbf{0.4375 } & 0.6745  & \textbf{0.3785 } & 0.9539  & \textbf{0.2866 } & 0.8084  & \textbf{0.3269 } \\
    \bottomrule
    \end{tabular}%
}
 \caption{Privacy Evaluation on CIFAR10.}
  \label{tab:privacy_cifar10}%
\end{table*}%
\begin{table*}[htbp]
  \centering
    \scalebox{0.8}{
    \begin{tabular}{|c|l|rr|rr|rr|rr|}
    \toprule
          &       & \multicolumn{2}{c|}{CNN6} & \multicolumn{2}{c|}{LeNet (Zhu)} & \multicolumn{2}{c|}{ResNet18} & \multicolumn{2}{c|}{VGG11} \\
    \midrule
    \multicolumn{1}{|l|}{Attack} & Metric  & \multicolumn{1}{l}{Baseline} & \multicolumn{1}{l|}{DGP} & \multicolumn{1}{l}{Baseline} & \multicolumn{1}{l|}{DGP} & \multicolumn{1}{l}{Baseline} & \multicolumn{1}{l|}{DGP} & \multicolumn{1}{l}{Baseline} & \multicolumn{1}{l|}{DGP} \\
    \midrule
    \multirow{2}[2]{*}{Rob attack} & LPIPS & 0.0212  & \textbf{0.4726 } & 0.0322  & \textbf{0.3782 } & 0.0297  & \textbf{0.4643 } & 0.0248  & \textbf{0.4677 } \\
          & SSIM  & 0.9298  & \textbf{0.1177 } & 0.9273  & \textbf{0.2257 } & 0.9322  & \textbf{0.1474 } & 0.9157  & \textbf{0.1209 } \\
    \midrule
    \multirow{2}[2]{*}{IG attack} & LPIPS & 0.0521  & \textbf{0.2996 } & 0.1042  & \textbf{0.3163 } & 0.0021  & \textbf{0.3747 } & 0.0297  & \textbf{0.3237 } \\
          & SSIM  & 0.7551  & \textbf{0.3792 } & 0.6849  & \textbf{0.3781 } & 0.9522  & \textbf{0.2550 } & 0.8057  & \textbf{0.2991 } \\
    \bottomrule
    \end{tabular}%
}
 \caption{Privacy Evaluation on CIFAR100.}
  \label{tab:privacy_cifar100}%
\end{table*}%



\subsection{Defenses under attacks with generative GIAs.}
For a comprehensive privacy evaluation, we assess existing defenses against GAN-based generative GIAs. We select GGL, the state-of-the-art generative GIA, and maintain its original strongest configuration, utilizing ResNet18 on ImageNet. 
Fig.~\ref{fig:GGL} shows that DGP has a better visualization effect. This is because {pruning Top-$k_1$ gradient elements 
(in DGP) will} confuse GGL's inference of some data labels, making the GAN-generated relevant data differ from the original data significantly.
\begin{figure}[t!]
    \centering
    \includegraphics[width=0.45\textwidth]{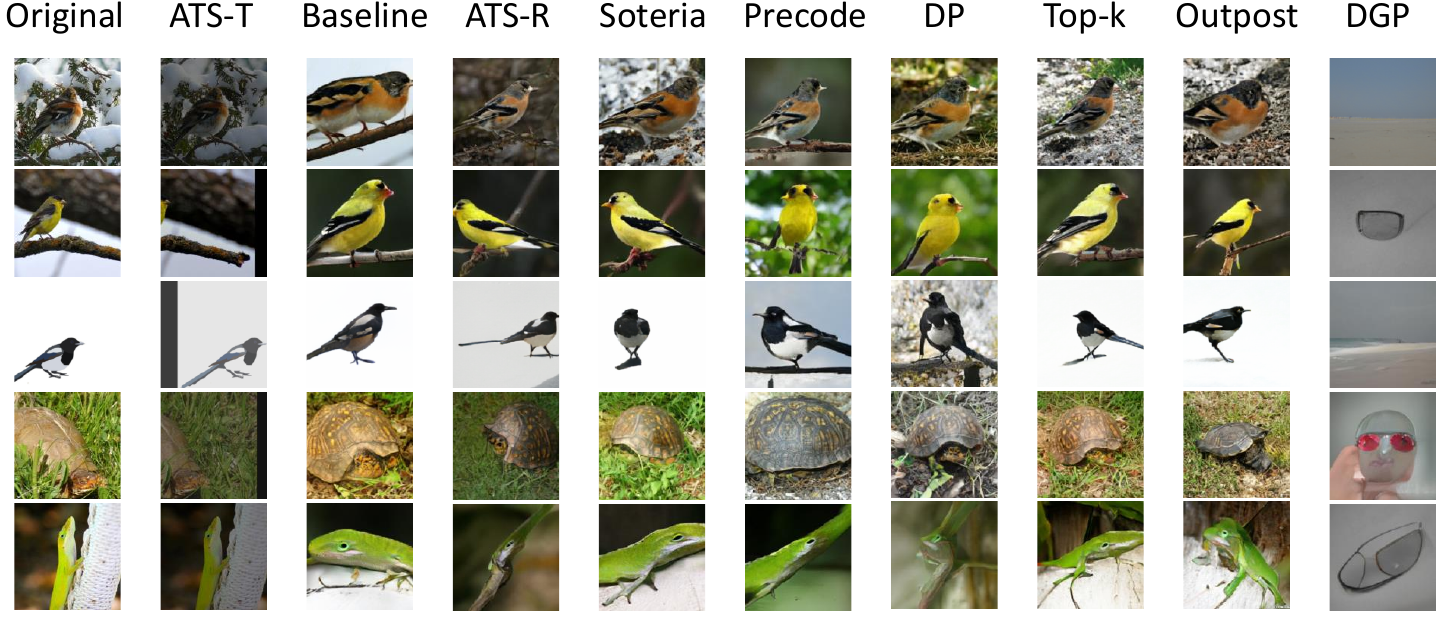}
        \centering
\caption{Reconstruction data visualization under GGL attack on ImageNet.}
  \label{fig:GGL}
\end{figure}
 
  

\subsection{Evaluation of the trade-off between privacy and accuracy for high-pruning rate Top-$k$}
Set Top-$k$ pruning rate for 95\%, DGP pruning rate for 80\%,  we compare their privacy-accuracy trade-offs.
As shown in Fig.~\ref{fig:topkDGP}, Fig.~\ref{fig:rob_topkDGP} and Tab.~\ref{tab:topkDGP}, with similar privacy protection, Top-$k$ is more likely to cause model performance degradation.
\begin{figure}[t!]
    \centering
    \includegraphics[width=0.4\textwidth]{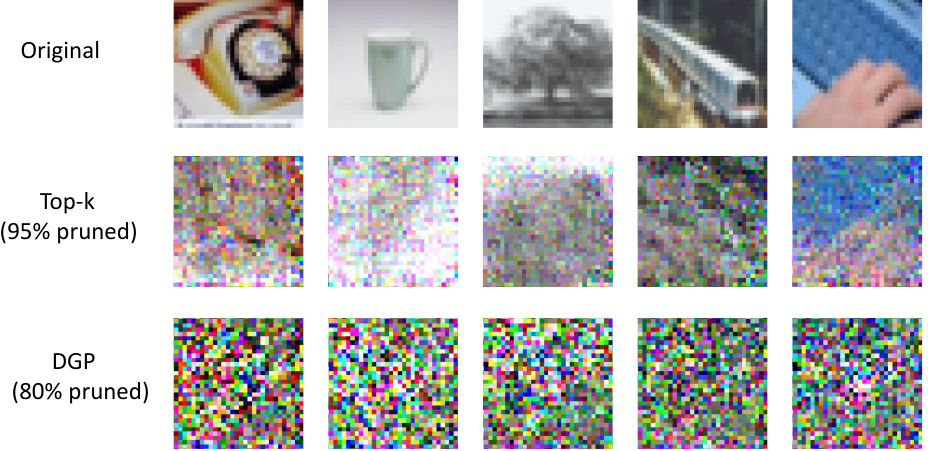}
        \centering
\caption{Reconstruction data visualization under IG attack on CIFAR100.}
  \label{fig:topkDGP}
\end{figure}
\begin{figure}[t!]
    \centering
    \includegraphics[width=0.45\textwidth]{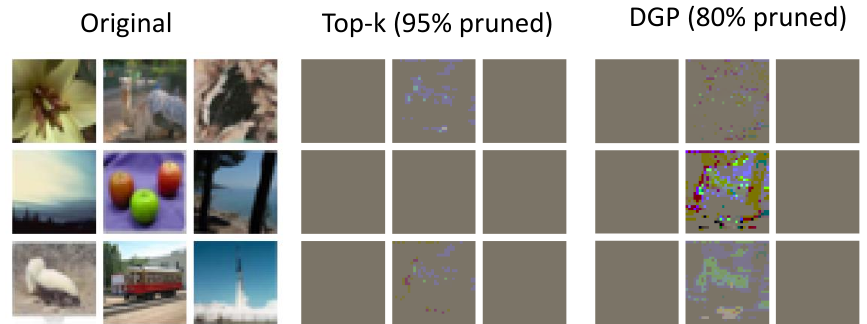}
        \centering
\caption{Reconstruction data visualization under Rob attack on CIFAR100.}
  \label{fig:rob_topkDGP}
\end{figure}
\begin{figure}[t!]
  \centering
  \subfigure[IG attack, batchsize=2]{
   \includegraphics[width=0.43\textwidth]{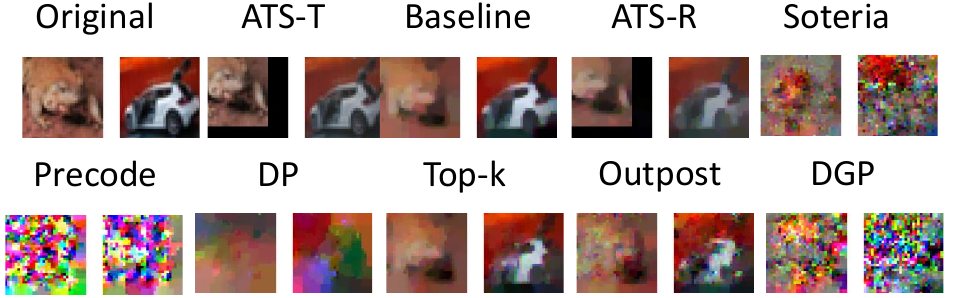} }
  \subfigure[IG attack, batchsize=4]{
\includegraphics[width=0.43\textwidth]{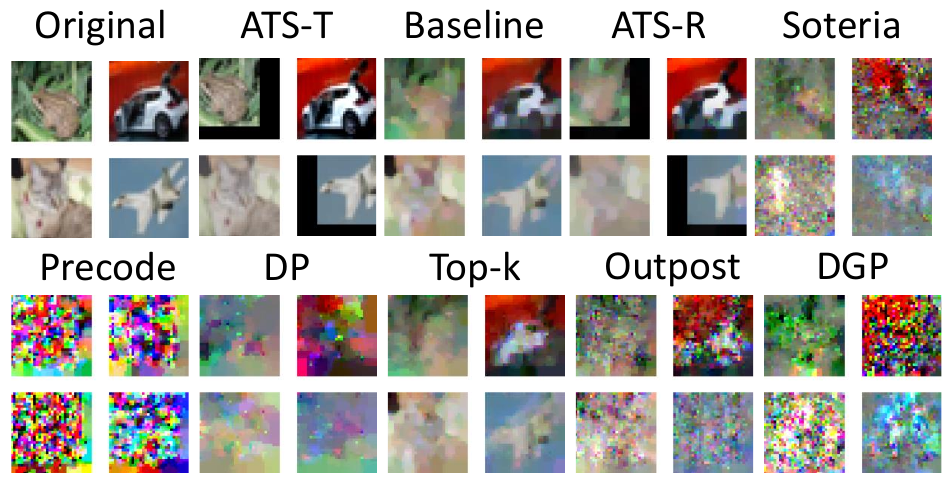}}
\caption{Reconstruction data visualization under IG attack with different batchsizes on  CIFAR10.}
     \label{fig:2_batchsize}
\end{figure}
\begin{figure} [t!] {
\centering
\subfigure{Rob attack, batchsize=2}{
\includegraphics[width=0.43\textwidth]{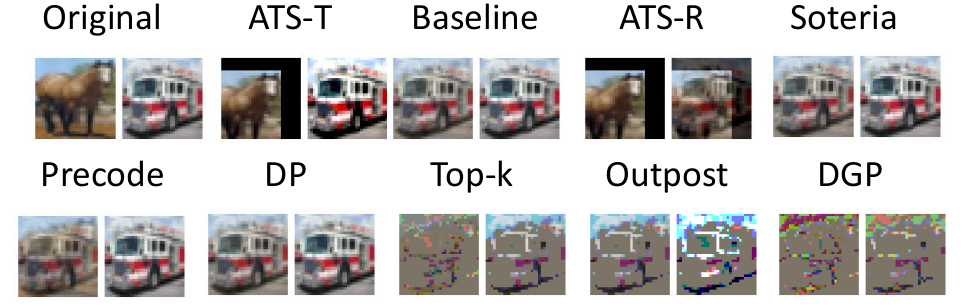} }
\centering
\subfigure{Rob attack, batchsize=4}{
\includegraphics[width=0.43\textwidth]{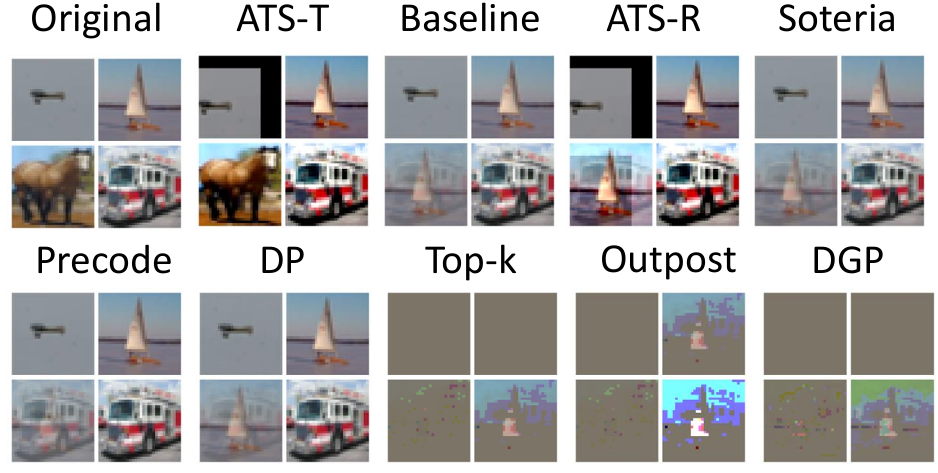} }
\caption{Reconstruction data visualization under Rob attack with different batchsizes on  CIFAR10.}
\label{fig:4_batchsize}
}
\end{figure}
\begin{table}[htbp]
    \vspace{-1em}
  \centering
   \scalebox{0.8}{
    \begin{tabular}{lrrrr}
    \toprule
          & \multicolumn{1}{l}{ResNet18} & \multicolumn{1}{l}{LeNet~(Zhu)} & \multicolumn{1}{l}{CNN6} & \multicolumn{1}{l}{VGG11} \\
    \midrule
    Top-$k$~(95\%) & 73.68\% & 26.75\% & 45.19\% & 68.07\% \\
    DGP~(80\%) & \textbf{74.04}\% & \textbf{32.42}\% & \textbf{47.24}\% & \textbf{68.60}\% \\
    \bottomrule
    \end{tabular}%
    }
     \caption{Model performance on CIFAR100.}
    \vspace{-1em}
  \label{tab:topkDGP}%
\end{table}%

\subsection{Defenses under attacks with different batch sizes}
To better evaluate privacy protection, we implement IG attack and Rob attack with different batchsizes. Fig.~\ref{fig:2_batchsize} and Fig.~\ref{fig:4_batchsize} show that our method protect privacy against IG and Rob attacks better than recent works. In particular, our method can comprehensively defend against gradient inversion attacks, while Top-$k$ and Outpost offer limited privacy protection against IG attack, and Soteria, ATS, Precode cannot defend against Rob attack.

\section{The framework of ADGP}
As shown in Fig.~\ref{fig:frame}, ADGP is achieved by randomly selecting a user, who broadcasts binary matrix $\mathcal{I}$ to all other users. Each user then only transmits gradient parameters whose locations belong to $\mathcal{I}$. 
\begin{figure}[htbp]
    \centering
    \includegraphics[width=8cm]{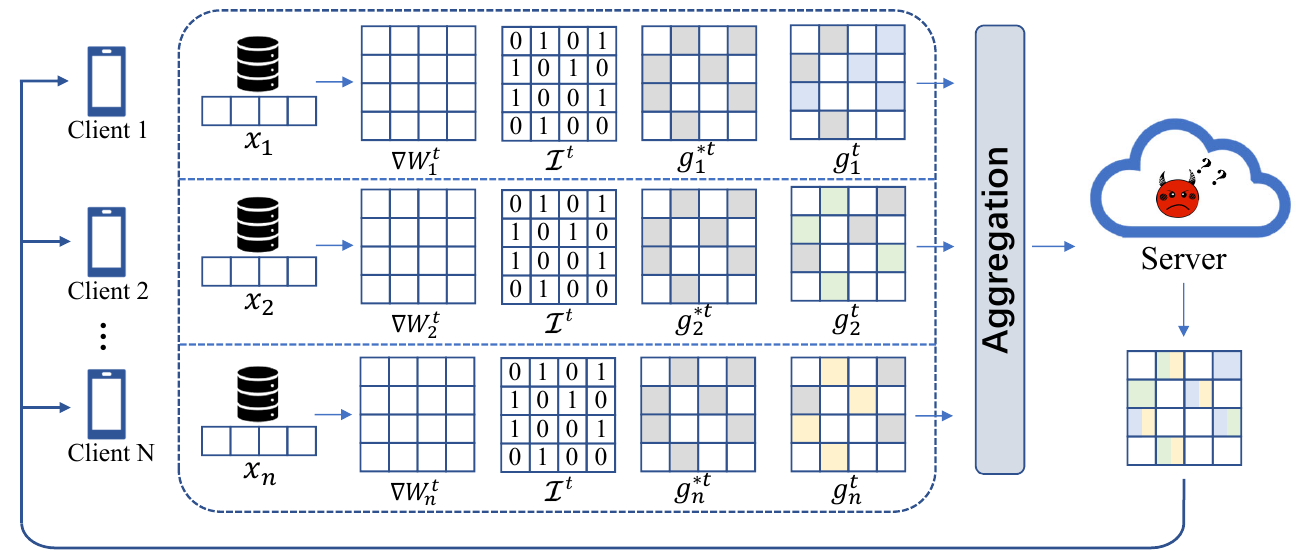}
\caption{The $t$-th iteration model update process, where $\mathbf{g}^{*t}$ represents the gradient parameters whose position belong to $\mathcal{I}$ in $t$-th iteration.}
    \label{fig:frame}
\end{figure}

\end{document}